\documentclass[lettersize,journal]{IEEEtran}
\usepackage{amsmath,amsfonts}
\usepackage{algorithm}
\usepackage{array}
\usepackage[caption=false,font=normalsize,labelfont=sf,textfont=sf]{subfig}
\usepackage{textcomp}
\usepackage{stfloats}
\usepackage{url}
\usepackage{verbatim}
\usepackage{graphicx}
\usepackage{cite}
\usepackage{amssymb,amsthm}
\usepackage{stfloats,color}
\usepackage{wasysym,epsfig}
\usepackage{enumerate,algpseudocode,balance,bm,nicefrac}
\usepackage{booktabs} 
\usepackage{multirow}
\usepackage{epstopdf}
\allowdisplaybreaks

\newtheorem{lemma}{Lemma}
\newtheorem{thm}{\bf Theorem}
\newtheorem{assumption}{\bf Assumption}

\newtheorem{corollary}{Corollary}

\newtheorem{problem}{Problem}
\theoremstyle{remark}
\theoremstyle{remark}\newtheorem{remark}{Remark}

\begin{document}


\title{\Large \bf Data-driven Polytopic Output Synchronization of Heterogeneous\\ Multi-agent Systems from Noisy Data}

\author{Yifei Li, Wenjie Liu, Jian Sun,~\IEEEmembership{Senior Member,~IEEE}, Gang Wang,~\IEEEmembership{Senior Member,~IEEE},\\
	Lihua Xie,~\IEEEmembership{Fellow,~IEEE}, and
	Jie Chen,~\IEEEmembership{Fellow,~IEEE}
\thanks{The work was supported in part by the National Key R\&D Program of China under Grant 2022ZD0119302, and the National Natural Science Foundation of China under Grants 62173034, 61925303, 62088101, U20B2073. 
	\emph{(Corresponding author: Gang Wang.)}
}
	\thanks{Yifei Li, Wenjie Liu, Jian Sun, and Gang Wang are with the National Key Lab of Autonomous Intelligent Unmanned Systems, Beijing Institute of Technology, Beijing 100081, China, and also with the Beijing Institute of Technology Chongqing Innovation Center, Chongqing 401120, China (e-mail: liyifei@bit.edu.cn; liuwenjie@bit.edu.cn; sunjian@bit.edu.cn; gangwang@bit.edu.cn).
		
Lihua Xie is with the Centre for Advanced Robotics Technology  Innovation (CARTIN), School of Electrical and Electronic Engineering, Nanyang Technological University, Singapore (e-mail:
elhxie@ntu.edu.sg).
	
	Jie Chen is with the Department of Control Science and Engineering, Tongji University, Shanghai 201804, China, and also with the National Key Lab of Autonomous Intelligent Unmanned Systems, Beijing Institute of Technology, Beijing 100081, China 	
	(e-mail: chenjie@bit.edu.cn).
}
}



\allowdisplaybreaks

\maketitle

\begin{abstract}

This paper proposes a novel approach to addressing the output synchronization problem in unknown heterogeneous multi-agent systems (MASs) using noisy data. Unlike existing studies that focus on noiseless data, we introduce a distributed data-driven controller that enables all heterogeneous followers to synchronize with a leader's trajectory. To handle the noise in the state-input-output data, we develop a data-based polytopic representation for the MAS. We tackle the issue of infeasibility in the set of output regulator equations caused by the noise by seeking approximate solutions via constrained fitting error minimization. This method utilizes measured data and a noise-matrix polytope to ensure near-optimal output synchronization. Stability conditions in the form of data-dependent semidefinite programs are derived, providing stabilizing controller gains for each follower. The proposed distributed data-driven control protocol achieves near-optimal output synchronization by ensuring the convergence of the tracking error to a bounded polytope, with the polytope size positively correlated with the noise bound. Numerical tests validate the practical merits of the proposed data-driven design and theory. 

\end{abstract}


\section{Introduction}

The field of distributed control, particularly consensus control of multi-agent systems (MASs), has garnered significant attention in recent decades. Numerous research efforts have focused on achieving state consensus in homogeneous MASs, where all agents share identical dynamics, as evidenced in e.g., \cite{Olfati2007, Ren2007, Li2013,cheng2023,you2013consensus,zhou2022lyapunov
} and their associated references. However, in real-world scenarios, the presence of inevitable variations between agents and system uncertainties stemming from physical characteristics introduce heterogeneity in the dynamics of MASs. Consequently, there is a growing need to investigate the problem of output synchronization in heterogeneous MASs, which consist of agents with different dimensions and dynamics and find extensive applications.

A well-documented approach for addressing the output synchronization problem is based on the internal model principle \cite{Wieland2011}. This principle has been widely applied to design protocols for output synchronization in various settings, including heterogeneous linear MASs \cite{Su2012, Hong2013,li2016distributed}, nonlinear MASs \cite{Isidori2014}, and their generalizations with additional considerations \cite{Jiang2023,Liu2018,Gao2016}.

However, these protocols rely on knowledge of the system dynamics model of each agent, rendering them inapplicable when first-principle models are unavailable or system identification is time-consuming or inaccurate. To overcome this challenge, several results have explored model-free approaches based on reinforcement learning (RL) \cite{Kiumarsi2017,Jiang2023rl,Chen2023}. Nevertheless, model-free RL-based approaches often require a large amount of training data, demanding considerable computing resources. Alternatively, the fundamental lemma introduced by \emph{Willems et al.} in \cite{Willems2005} provides an alternative avenue for designing controllers for unknown systems by using pre-collected data, such as input, output, and/or state data. This line of research has gained significant attention due to its advantages in terms of theoretical certification and computational tractability compared to other data-driven approaches \cite{Hou2013,vanwaarde2023informativity,Markovskydata}. Moreover, a data-based representation for linear time-invariant systems has been proposed in \cite{Persis2020}, which enables the design of stabilizing controllers using data-dependent linear matrix inequalities.

The fundamental lemma has lately found applications in diverse areas of data-driven control design and analysis, including data-driven model predictive control \cite{Coulson2019, Berberich2021, Liu2022data}, 
data-enabled policy optimization \cite{Zhao2023dataenabled},
robust control \cite{Waarde2022,Rueda2021data},
event-triggered control \cite{Wang2023jas,Qi2022,Liu2022self,Wang2022self}, 
 and distributed control of network systems \cite{Baggio2021, Jiao2021, li2022ddconsensus}, among others.

The problem of data-driven control design and analysis for output synchronization in unknown heterogeneous MASs remains unexplored. While some progress has been made, such as the work presented in \cite{Jiao2021} where the process noise during offline data collection was assumed to be measurable and perfectly known, the general case of having unknown noise has not been addressed. This limitation arises from the fact that process noise cannot be accurately measured in practical scenarios. Consequently, there is a pressing need to revisit the data-driven output synchronization problem in heterogeneous MASs, considering the presence of unknown noise. Moreover, accurate system identification is hindered by noisy data, necessitating the development of robust data-based methods for controller design that can handle uncertainty. The fundamental premise of such methods is to impose reasonable constraints on the noise.

In the field of data-driven control, two primary approaches are commonly used to model noise in the data acquisition phase, namely zonotopic constraints and quadratic constraints. The former approach, utilizing data-driven zonotopic reachability analysis \cite{
	Alanwar2021}, has led to the development of robust data-driven predictive controllers as demonstrated in \cite{Alanwar2022, russo2022tube}. On the other hand, the latter approach, introduced as a general framework in \cite{Waarde2022} leveraging the matrix \emph{S}-lemma, has been widely employed in designing controllers from data subject to noise modeled by quadratic constraints. Notable applications include 
robust event-triggered control \cite{Wang2023jas},
	and distributed control \cite{li2022ddconsensus}. Furthermore, the feasibility of the output regulator equations, a crucial component for achieving output synchronization using the internal model principle, is compromised by noisy data. To the best of our knowledge, no data-driven methods have been reported for solving the output regulator equations from noisy data.

Motivated by these observations, this paper aims to develop data-driven polytopic controllers for output synchronization of unknown heterogeneous MASs using noisy data obtained offline. The first step involves describing the noise using polytopes, which serves as the basis for a novel data-based polytopic representation of MASs. To address the infeasibility issue of output regulator equations, we seek approximate solutions by tackling a norm minimization problem that incorporates noise using matrix polytopes. Building on the polytopic MAS representation, we derive sufficient conditions for stabilizing  feedback controllers in the form of data-dependent semidefinite programs. Further, we consider a new synchronization measure, termed $\Delta$-optimal output synchronization, which accounts for regulator equation errors in the presence of noisy data and where $\Delta$ is proportional to the size of noise. We demonstrate that the proposed data-driven control protocol achieves $\Delta$-optimal output synchronization and exhibits robustness against the noise in data collected offline. This is achieved by ensuring the convergence of tracking error to a $\Delta$-bounded polytope. Notably, the synchronization recovers exact output synchronization when the noise vanishes.


The contributions of this paper are summarized as follows:
\begin{enumerate}[c1)]
	\item We propose a data-based polytopic representation for heterogeneous MASs based on state-input-output data corrupted by bounded noise;
	\item We derive approximate solutions to the data-driven output synchronization problem by minimizing the norm of fitting error matrix polytopes; and,
	\item We establish the near-optimality and robustness of the data-driven polytopic design of distributed output synchronization controllers against noise.
\end{enumerate}

%


 \emph{Notation.} We adopt the following notation conventions throughout the paper. The set of non-negative integers (real numbers) is denoted by $\mathbb{N}$ ($\mathbb{R}$). The sets of $n$-dimensional real vectors and $n \times m$ real matrices are represented by $\mathbb{R}^n$ and $\mathbb{R}^{n \times m}$, respectively. For a vector $x \in \mathbb{R}^n$, the notation $x > 0$ indicates that each entry of $x$ is positive. The symbol $(\cdot)^\top$ denotes the transpose operation, while $\otimes$ represents the Kronecker product. The identity matrix of appropriate dimensions is denoted by $I$, and the zero matrix is denoted by $\mathbf{0}$. The Frobenius norm of a real matrix $X \in \mathbb{R}^{n \times m}$ is denoted by $\|X\|_{F}$. A symmetric matrix $P$ is said to be positive definite (positive semi-definite) if $P \succ \mathbf{0}$ ($P \succeq \mathbf{0}$).  The expression ${\rm{diag}}_{i=1}^{N}\{a_i\}$ represents a diagonal matrix with $a_1, a_2, \ldots, a_N$ as its main diagonal elements.

\section{Preliminaries and Problem Formulation}
\label{section2}
\subsection{Graph Theory}

Consider a weighted graph ${\mathcal{G}}=({\mathcal{V}}, {\mathcal{E}})$ that represents the interactions among a set of agents. The graph consists of two components: a nonempty set of nodes ${\mathcal{V}}=\{v_1,\ldots, v_N \}$ and a set of edges ${\mathcal{E}} \subseteq {\mathcal{V}}\times {\mathcal{V}}$. An element in ${\mathcal{E}}$, denoted as $(v_i, v_j)$, represents a link from node $v_j$ to node $v_i$.

The adjacency matrix $\mathcal{A}=[a_{ij}]\in \mathbb{R}^{N\times N}$ is defined such that 
$a_{ij}>0$ if $(v_j, v_i)\in \mathcal{E}$, and $a_{ij}=0$ otherwise. The in-degree of node $v_i$ in graph $\mathcal{G}$ is given by $d_i=\sum_{j=1}^{N}a_{ij}$, and can be represented by the diagonal matrix $D={\rm{diag}}_{i=1}^{N}\{d_i\}$. The Laplacian matrix  $\mathcal{L}=[l_{ij}]\in \mathbb{R}^{N\times N}$ associated with $\mathcal{G}$ is defined as $\mathcal{L}=D-\mathcal{A}$.

The neighbor set of node $v_i$ is denoted as $\mathcal{N}_i=\{j\in {\mathcal{V}}|(i,j)\in {\mathcal{E}}\}$. An extended graph is represented by $\bar{\mathcal{G}}=(\bar{\mathcal{V}}, \bar{\mathcal{E}})$, where $\bar{\mathcal{V}}=\mathcal{V}\cup v_0$, and $v_0$ corresponds to the node associated with the leader. The set $\bar{\mathcal{E}}$ includes all the arcs in ${\mathcal{E}}$ as well as the arcs between $v_0$ and ${\mathcal{E}}$.

A graph $\bar{\mathcal{G}}$ is said to contain a directed spanning tree if there exists a node, known as the root, from which every other node in $\bar{\mathcal{V}}$ can be reached through a directed path. The pinning matrix $G = {\rm{diag}}_{i=1}^{N}\{g_i\}$ describes the accessibility of the leader node $v_0$ to the remaining nodes $v_i\in \mathcal{V}$. The pinning gain $g_i > 0$ if $(v_0,v_i)\in \bar{\mathcal{E}}$, and $g_i = 0$ otherwise.

\subsection{Output Synchronization of Discrete-time MASs}
\label{Sec2b}

Consider a discrete-time heterogeneous leader-following MAS consisting of a leader indexed by $0$ and $N$ followers indexed by $1,2,\ldots,N$. The dynamics of follower $i\in\{1,2,\ldots,N\}$ is described by
\begin{equation}
	\begin{split}
		\label{mas}
		{x}_{i}(t+1)&=\bar{A}_{i} x_{i}(t)+\bar{B}_{i} u_{i}(t)\\
		y_i(t)&=\bar{C}_ix_{i}(t),\quad \forall t\in\mathbb{N}
	\end{split}
\end{equation}
where $x_i(t)\in \mathbb{R}^{n_i}$ represents the state, $u_i(t)\in \mathbb{R}^{p_i}$ denotes the control input, and $y_i(t)\in \mathbb{R}^{q}$ is the measurement output.

In this paper, the true system matrices $\bar{A}_i\in\mathbb{R}^{n_i\times n_i}$, $\bar{B}_i\in\mathbb{R}^{n_i\times p_i}$, and $\bar{C}_i\in\mathbb{R}^{q\times n_i}$ are assumed unknown. 

The dynamics of the leader is given by
\begin{equation}
	\label{leader}
	\begin{split}
		{x}_{0}(t+1)&={S} x_{0}(t)\\
		{y}_{0}(t)&={H} x_{0}(t)
	\end{split}
\end{equation}
where $x_{0}(t)\in \mathbb{R}^{n_0}$ and $y_{0}(t)\in \mathbb{R}^{q}$ represent the state and output of the leader, respectively. The leader's system matrices $S$ and $H$ are assumed real, constant, and known. The pair $(S,H)$ is assumed to be observable.

The dynamics and state dimensions are allowed to differ across agents, while the output dimensions must be identical for synchronization. The objective is to synchronize the outputs of all followers with that of the leader by implementing a distributed  feedback control protocol for the MAS described by \eqref{mas}-\eqref{leader}, such that $\lim_{t\to \infty}\|y_i(t)-y_0(t)\|=0$ holds for all $i\in\{1,2,\ldots,N\}$. To address this problem, the following assumptions are made.

\begin{assumption}[Communication topology]
	\label{graph}
	The graph $\bar{\mathcal{G}}$ contains a directed spanning tree with the leader node as the root.
\end{assumption}

\begin{assumption}[Stabilizability and detectability]
\label{ABC}
The pair $(\bar{A}_i,\bar{B}_i)$ is stabilizable, and $(\bar{C}_i,\bar{A}_i)$ is detectable for all $i\in\{1,2,\ldots,N\}$. 
\end{assumption}

\begin{assumption}[{Oscillating leader dynamics}]
\label{pole}
The leader dynamics $S$ has all its poles on the unit circle and non-repeated.
\end{assumption}

Regarding the assumptions, we have a remark.
\begin{remark}
	Assumptions~\ref{graph}--\ref{pole} are standard for achieving output synchronization in linear heterogeneous MASs and have been utilized in several existing results, such as \cite{Kiumarsi2017, Jiao2021}. It is important to note that the state of the leader is bounded and does not converge to zero under Assumption~\ref{pole}.

\end{remark}

Based on these assumptions, we consider a distributed  feedback control protocol for each follower in \eqref{mas} as follows
\begin{equation}
\label{controller}
		u_{i}(t)=K_i[x_i(t)-\Pi_i\eta_i(t)]+\Gamma_i\eta_i(t)
\end{equation}
where $K_i\in \mathbb{R}^{p_i\times n_i}$ is the feedback gain matrix to be designed, and $\Pi_i\in \mathbb{R}^{n_i\times n_0}$ and $\Gamma_i\in \mathbb{R}^{p_i\times n_0}$ are the solutions to output regulator equations given by
\begin{equation}
	\label{regulator}
	\begin{split}
	\bar{A}_i\Pi_i+\bar{B}_i\Gamma_i&=\Pi_i S\\
	\bar{C}_i\Pi_i&=H.
	\end{split}
\end{equation}

An observer $\eta_i(t)\in \mathbb{R}^{n_0}$ is employed to estimate the state of the leader, governed by the following distributed observer
\begin{equation}
	\label{observer}
	\begin{split}
		\eta_i(t&+1)=S\eta_i(t)+(1+d_i+g_i)^{-1}F\\
		&\times\Big[\sum_{j=1}^{N}a_{ij}(\eta_j(t)-\eta_i(t))+g_i(x_0(t)-\eta_i(t))\Big]
	\end{split}
\end{equation}
where $d_i$ and $g_i$ are the in-degree and pinning gain of node $i$, respectively, and
$F\in \mathbb{R}^{n_0\times n_0}$ is a gain matrix to be designed. 

Next, we define the observer's estimation error $\delta_i(t):=\eta_i(t)-x_0(t)$. It follows from \eqref{leader} and \eqref{observer} that the dynamics of $\delta_i(t)$ satisfies 
\begin{equation}
\label{delta}
\begin{split}
	\delta_i(t+1)&={S}\delta_i(t)+(1+d_i+g_i)^{-1}F\\
	&\quad\times \Big[\sum_{j=1}^{N}a_{ij}(\delta_j(t)-\delta_i(t))-g_i\delta_i(t)\Big].
\end{split}
\end{equation}
For the entire system, \eqref{delta} can be expressed in a compact form as follows
\begin{equation}
	\label{eq:esti:entire}
	\delta(t+1)=\big[I_N\otimes S- (I_N+D+G)^{-1}(\mathcal{L}+G)\otimes F\big]\delta(t)
\end{equation}
where $\delta(t)=[\delta_1^\top(t)~ \delta_2^\top(t)~\cdots~\delta_N^\top(t)]^\top$.

Next, we introduce the virtual tracking error $\xi_i(t):= x_i(t) - \Pi_i\eta_i(t)$. By substituting \eqref{mas}, \eqref{controller}, and \eqref{observer} into the definition of $\xi_i(t)$, we obtain its dynamics as follows
\begin{equation}
\label{eq:auxerr:dyn}
	\xi_i(t+1)=(\bar{A}_i+\bar{B}_iK_i)\xi_i(t)+(1+d_i+g_i)^{-1}\Pi_iFz_i(t)
\end{equation}
where $z_i(t)=\sum_{j=1}^{N}a_{ij}(\delta_i(t)-\delta_j(t))+g_i\delta_i(t)$.

When the true system matrices $(\bar{A}_i, \bar{B}_i, \bar{C}_i)$ are known, the next lemma provides necessary and sufficient conditions for achieving output synchronization, see e.g., \cite[Theorem 1]{Kiumarsi2017}.

\begin{lemma}[{\cite[Theorem 1]{Kiumarsi2017}}]
\label{lem1}
Suppose Assumptions~\ref{graph}-\ref{pole} hold. The output synchronization of the MAS \eqref{mas}-\eqref{leader} is achieved for all initial conditions under the distributed  feedback control strategy \eqref{controller}-\eqref{observer}, if and only if there exist matrices $F$ and $K_i$ such that $I_N\otimes S- (I_N+D+G)^{-1}(\mathcal{L}+G)\otimes F$ and $\bar{A}_i + \bar{B}_iK_i$ are Schur stable for all $i\in\{1,2,\ldots,N\}$.
\end{lemma}


\subsection{Pre-collecting Noisy Data}
\label{section2c}

In order to address the challenge of unknown system matrices for each follower, we propose a distributed data-driven approach. Specifically, we excite each follower with some control inputs, obtaining a set of data denoted by $\mathbb{D}_i=\{(x_i(T),u_i(T),y_i(T)): T\in \{0,1,\ldots,\rho\}\}$ for each follower $i$. The set $\mathbb{D}_i$ consists of state, input, and output measurements, and is obtained through an open-loop experiment on the following perturbed system
\begin{equation}
\label{masnoisy}
	\begin{split}
		x_i(T+1)&=\bar{A}_{i}x_i(T)+\bar{B}_{i}u_i(T)+w_i(T)\\
		y_i(T)&=\bar{C}_ix_i(T)+v_i(T)
	\end{split}
\end{equation}
where $w_i(T)\in\mathbb{R}^{n_i}$ and $v_i(T)\in\mathbb{R}^{q}$ represent {unknown} process and measurement noise, respectively. These noises satisfy the following assumption.

\begin{assumption}[Polytopic noise]
	\label{as:noise}
	For every $T$ and $i$, the process noise $w_i(T)$ and measurement noise $v_i(T)$ belong respectively to polytopic sets $\mathcal{P}_{w_i}$ and $\mathcal{P}_{v_i}$ given by
\begin{equation*}
	\begin{split}
	\mathcal{P}_{w_i}&=\Big\{w_i\big|w_i=\sum_{k=1}^{\gamma_{w_i}}\beta^{(k)}_{w,i}\hat{w}_{i}^{(k)},\beta^{(k)}_{w,i}\geq0, \sum_{k=1}^{\gamma_{w_i}}\beta^{(k)}_{w,i}=1 \Big\}\\
\mathcal{P}_{v_i}&=\Big\{v_i\big|v_i=\sum_{k=1}^{\gamma_{v_i}}\beta^{(k)}_{v,i}\hat{v}_{i}^{(k)},\beta^{(k)}_{v,i}\geq0, \sum_{k=1}^{\gamma_{v_i}}\beta^{(k)}_{v,i}=1 \Big\}
	\end{split}
\end{equation*}
where $\hat{w}_{i}^{(k)}$ and $\hat{v}_{i}^{(k)}$ represent the $k$-th vertex of polytopes $\mathcal{P}_{w_i}$ and $\mathcal{P}_{v_i}$, respectively, and  $\gamma_{w_i}$ and $\gamma_{v_i}$ denote the number of vertices.
\end{assumption}

To store the collected data, we define the following matrices per agent
\begin{align*}
	 	X_{i+} &:=\left[x_i(1) \; x_i(2) \; \cdots \; x_i(\rho)\right]\\
	 	X_i& :=\left[x_i(0) \; x_i(1) \; \cdots \; x_i(\rho-1)\right]\\
	U_i& :=\left[u_i(0) \; u_i(1) \; \cdots \; u_i(\rho-1)\right]\\
	 Y_i& :=\left[y_i(0) \; y_i(1) \; \cdots \; y_i(\rho-1)\right].
\end{align*}

The unknown process noise of length $\rho$ is denoted as $\{w_i(T)\}_{T=0}^{\rho-1}$. Consequently, the stacked matrix per agent, $W_i=[w_i(0) \; w_i(1) \; \cdots \; w_i(\rho-1)]$, belongs to the matrix polytope $\mathcal{M}_{W_i}$, 
 given by
\begin{equation}
	\label{M_poly}
	\mathcal{M}_{W_i}=\Big\{W_i\Big|W_i=\sum_{k=1}^{\gamma_{w_i}\rho}\beta_{W,i}^{(k)}\hat{W}_{i}^{(k)},\beta_{W,i}^{(k)}\geq0, \sum_{k=1}^{\gamma_{w_i}\rho}\beta_{W,i}^{(k)}=1 \Big\}
\end{equation} 
which results from the concatenation of multiple disturbance polytopes $\mathcal{P}_{w_i}$ as follows
\begin{equation}\label{eq:whati}
	\begin{split}
		\hat{W}_i^{(1+(k-1)\rho)}&=\big[\hat{w}_{i}^{(k)} \quad \mathbf{0}_{n_i\times (\rho-1)} \big]\\
	\hat{W}_i^{(m+(k-1)\rho)}&=\big[\mathbf{0}_{n_i\times (m-1)}\quad \hat{w}_{i}^{(k)}\quad \mathbf{0}_{n_i\times (\rho-m)} \big]	\\
	\hat{W}_i^{(\rho+(k-1)\rho)}&=\big[\mathbf{0}_{n_i\times (\rho-1)}\quad \hat{w}_{i}^{(k)} \big]
	\end{split}
\end{equation}
for each $k \in \{1,2,\ldots,\gamma_{w_i}\}$, $m \in\{2,3, \ldots,\rho-1\}$, and $i \in \{1,2,\ldots,N\}$.

Similarly, upon denoting the sequence of measurement noise $\{v_i(T)\}_{T=0}^{\rho-1}$ as $V_i=[v_i(0) \; v_i(1) \; \cdots \; v_i(\rho-1)]$, we can deduce that $V_i$ belongs to the following matrix polytope 
\begin{equation}
	\mathcal{M}_{V_i}=\Big\{V_i\Big|V_i=\sum_{k=1}^{\gamma_{v_i}\rho}\beta_{V,i}^{(k)}\hat{V}_{i}^{(k)},\beta_{V,i}^{(k)}\geq0, \sum_{k=1}^{\gamma_{v_i}\rho}\beta_{V,i}^{(k)}=1 \Big\}
\end{equation}
with the vertices $\hat{V}_{i}^{(k)}$ defined in the same way as $	\hat{W}_i^{(k)}$ in \eqref{eq:whati}.

\subsection{Problem Statement}

Having introduced the output synchronization of discrete-time heterogeneous MASs and the pre-collected data, we now formally state the problem of interest.

\begin{problem}[Data-driven output synchronization]
	\label{problem1}
	Given the state-input-output measurements $\mathbb{D}_i$ for each follower $i$, and under Assumptions \ref{graph}-\ref{pole}, the objective is to design a distributed controller of the form \eqref{controller}-\eqref{observer} such that approximate output synchronization of the heterogeneous MAS \eqref{mas}-\eqref{leader} is achieved for any initial states.
\end{problem}

The main challenge in addressing Problem \ref{problem1} lies in solving the output regulator equations in \eqref{regulator}, designing the controller gain $K_i$, and performing synchronization analysis without knowledge of the true system matrices, but using only available data. To tackle this challenge, we propose a data-driven polytopic reachability analysis technique in this paper, inspired by the zonotopic reachability analysis presented in \cite{Alanwar2021}.

\section{Distributed Data-driven Output Synchronization of Heterogeneous MASs}
\label{section3}

This section addresses the challenging problem of output synchronization in the unknown heterogeneous MAS \eqref{mas}-\eqref{leader}. Due to the presence of noisy data, achieving asymptotic output synchronization for the unknown heterogeneous MAS is impractical, unlike the model-based scenario depicted in Lemma \ref{lem1}. Instead, we propose achieving $\Delta$-optimal output synchronization by ensuring the stability of reachable error trajectories, which will be formalized in the following sections.

Our approach begins by introducing a data-based polytopic representation that characterizes an open-loop MAS using noisy data $\mathbb{D}_i$. Next, we address the output regulator equations by solving a data-dependent norm minimization problem. Subsequently, a polytopic controller is designed directly from the data. Leveraging this controller and the approximate solution to the output regulator equations, we propose a data-driven output synchronization algorithm and provide a proof of its UBB property, effectively addressing Problem \ref{problem1}. See Fig. \ref{fig:flow} for an illustration of the leader-following heterogeneous MAS architecture.

\begin{figure}[!htb]
	\centering
	\includegraphics[scale=0.38]{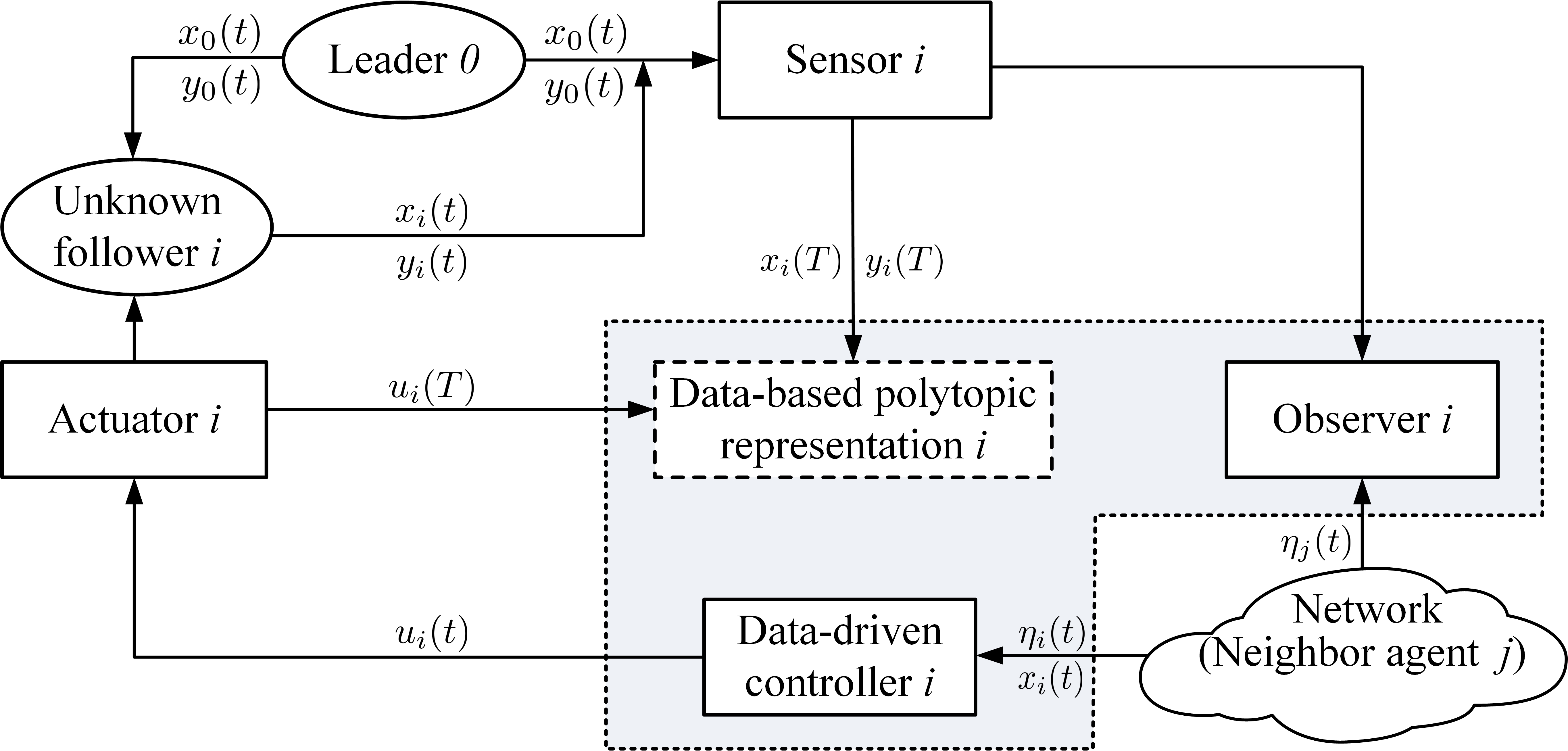}
	\caption{Distributed data-driven output synchronization.}
	\label{fig:flow}
\end{figure}

\subsection{Data-based Polytopic  Representation of MASs}
In this part, we assume the next rank condition on the richness of the data in $\mathbb{D}_i$.
\begin{assumption}
	\label{rank}
	The data matrix $\left[\begin{matrix}U_i \\X_i\end{matrix}\right]$ has full row rank for each $i\in\{1,2,\ldots,N\}$.
\end{assumption}

The verification of Assumption~\ref{rank} can be easily performed for a given dataset $\mathbb{D}_i$. Note that this rank condition can be satisfied if the measurements $(u_i(T), w_i(T))$ are persistently exciting of order $n_{i}+1$ \cite{Persis2020}. We now proceed with representing the heterogeneous MAS using data.

Since the actual realization of noise $(w_i(T), v_i(T))$ is unknown, there generally exist many systems $(A_i,B_i,C_i)$ that are consistent with the data $(X_i,X_{i+},U_i,Y_i)$. We denote this set by $\Sigma_i$ as follows:
\begin{equation}
	\begin{split}\label{eq:sigma}
		\Sigma_{i}:=\big\{(A_i&,\, B_i,\,C_i)| X_{i+}=A_iX_i+B_iU_i+W_i, \\
		&Y_i=C_iX_i+V_i, W_i\in\mathcal{M}_{W_i},V_i\in\mathcal{M}_{V_i}\big\}.
		\end{split}
\end{equation}

As mentioned in Sec.~\ref{section2c}, we assume knowledge of the polytopes $\mathcal{P}_{w_i}$ and $\mathcal{P}_{v_i}$, which bound the noise $w_i(T)$ and $v_i(T)$, along with their associated matrix polytopes $\mathcal{M}_{W_i}$ and $\mathcal{M}_{V_i}$, respectively. Our objective now is to compute a set $\mathcal{M}_i$ that provides an overapproximation of all possible $(A_i,B_i,C_i)$ consistent with the state-input-output data and given noise bounds. To achieve this, we construct a data-based polytopic representation of $(A_i,B_i,C_i)$ in the following lemma, inspired by \cite{Alanwar2022}. This representation yields a matrix polytope $\mathcal{M}_i\supseteq \Sigma_i$.

\begin{lemma}[Data-based polytopic representation of MASs]
	\label{representation}
	Suppose Assumptions~\ref{as:noise}-\ref{rank} hold. For each $i\in\{1,2,\ldots,N\}$, given state-input-output data $(X_i,X_{i+},U_i,Y_i)$ of the MAS \eqref{masnoisy}, the matrix polytope is defined as:
		\begin{equation}
		\label{M_AB}
		\begin{split}
			\mathcal{M}_{i}=\Big\{
			(\mathcal{M}_{{Z}_i},\mathcal{M}_{{C}_i})\Big|
			\mathcal{M}_{{Z}_i}&=(X_{i+}-\mathcal{M}_{W_i})\left[\begin{matrix}U_i \\X_i\end{matrix}\right]^{\dagger},\\
			\mathcal{M}_{{C}_i}&=(Y_{i}-\mathcal{M}_{V_i})X_i^{\dagger}\Big\}.
		\end{split}
	\end{equation}
This matrix polytope characterizes all matrices $(A_i, B_i,C_i)$ consistent with the data $(X_i,X_{i+},U_i,Y_i)$ and noise bounds, i.e., $ \Sigma_i \subseteq  \mathcal{M}_{i}$.
\end{lemma}

\begin{proof}
	Consider any $(A_i,B_i,C_i)\in \Sigma_{i}$. From \eqref{eq:sigma}, we can find $W_i\in \mathcal{M}_{W_i}$ and $V_i\in \mathcal{M}_{V_i}$ such that:
	\begin{equation}
		\begin{split}
		\label{AB1}
		A_iX_i+B_iU_i&=X_{i+}-W_i,\\
		C_iX_i&=Y_i-V_i.
		\end{split}
	\end{equation}
By using \eqref{M_poly}, we can express every $W_i\in\mathcal{M}_{W_i}$ uniquely as  $W_i=\sum_{k=1}^{\gamma_{w_i}\rho}\beta_{W,i}^{(k)}\hat{W}_{i}^{(k)}$ with coefficients $\{\beta_{W,i}^{(k)}\}$. Furthermore, multiplying $\left[\begin{matrix}U_i \\X_i\end{matrix}\right]^{\dagger}$ from the right on both sides of the first equation in \eqref{AB1} yields:
	\begin{equation}
		\label{AB2}
		Z_i:=[B_i~ A_i]=\Big(X_{i+}-\sum_{k=1}^{\gamma_{w_i}\rho}\beta_{W,i}^{(k)}\hat{W}_{i}^{(k)}\Big)\left[\begin{matrix}U_i \\X_i\end{matrix}\right]^{\dagger}.
	\end{equation}

Similarly, it can be easily deduced that:
	\begin{equation}
		\label{C1}
	C_i=\Big(Y_{i}-\sum_{k=1}^{\gamma_{v_i}\rho}\beta_{V,i}^{(k)}\hat{V}_{i}^{(k)}\Big)X_i^{\dagger}.
	\end{equation}
Hence, for every $(A_i,B_i,C_i)\in \Sigma_i$, there exist $\beta_{W,i}^{(k)}$, $k=1,2,\ldots,\gamma_{w_i}\rho$, and $\beta_{V,i}^{(k)}$, $k=1,2,\ldots,\gamma_{v_i}\rho$, such that $\sum_{k=1}^{\gamma_{w_i}\rho}\beta_{W,i}^{(k)}=1$ and $\sum_{k=1}^{\gamma_{v_i}\rho}\beta_{v,i}^{(k)}=1$, satisfying \eqref{AB2} and \eqref{C1}.
	Therefore, for every $(A_i,B_i,C_i) \in \Sigma_i$, it also holds that $(A_i,B_i,C_i) \in \mathcal{M}_{\Sigma_i}$ for $i=1,2,\ldots,N$ as defined in \eqref{M_AB}, which concludes the proof.
\end{proof}

It is worth noting that a prevalent approach in prior literature for modeling unknown, yet bounded noise is to use the energy form, typically in terms of a quadratic full-block bound; see, e.g., \cite{Persis2020,Berberich2021,Waarde2022,li2022ddconsensus,Wang2023jas}. In contrast, we propose a novel approach by describing unknown noise using polytopes. We then develop a data-based polytopic representation of MASs in Lemma~\ref{representation}, which allows us to characterize all possible system matrices. This representation paves the way for addressing Problem \ref{problem1} in the subsequent analysis. Notably, compared with the quadratic matrix inequality-based representation in \cite{Waarde2022,li2022ddconsensus,Wang2023jas}, our proposed polytopic representation maintains the simplicity and compactness, while providing a more precise characterization of the system and resulting in less conservative data-based stability conditions.

\subsection{Solution to Output Regulation Equations with Noisy Data}
\label{section3b}

In this subsection, we introduce a data-driven approach to solve the output regulator equations by minimizing the norm of noise-matrix polytopes. This approach yields approximate solutions $(\Pi_i,\Gamma_i)$ for each follower directly from noisy data, without requiring exact knowledge of the true system matrices $(\bar{A}_i,\bar{B}_i,\bar{C}_i)$.

We point out that the presence of $w_i(T)$ and $v_i(T)$ in data $\mathbb{D}_i$ prevents us from obtaining accurate solutions to the initial output regulator equations \eqref{regulator}. Thus, we define $\Delta_{i1}$ and $\Delta_{i2}$ as the errors in the regulator equations caused by the noise-corrupted data. As an intermediate step, we modify the output regulator equations \eqref{regulator} for any $(A_i,B_i,C_i)\in\mathcal{M}_i$ as follows:
	\begin{equation}
	\label{eq:regu:noise}
	\begin{split}
	\Delta_{i1}&=A_i\Pi_i+B_i\Gamma_i-\Pi_i S,\\
	\Delta_{i2}&=C_i\Pi_i-H.
	\end{split}
\end{equation}

Next, we make a crucial observation that directly finding the solutions $\Pi_i$ and $\Gamma_i$ from \eqref{eq:regu:noise} is infeasible due to the presence of unknown terms $\Delta_{i1}$ and $\Delta_{i2}$. Accordingly, we formulate the problem of determining the gains $\Pi_i$ and $\Gamma_i$ for any $(A_i,B_i,C_i)\in\mathcal{M}_i$ as an optimization problem that minimizes $\Delta_{i1}$ and $\Delta_{i2}$ with respect to a chosen norm. According to the data-based polytopic representation in Lemma~\ref{representation}, we define the following optimization problem:
\begin{equation}\label{eq:opt}
	\underset{\Pi_i,\Gamma_i}{\min}
	~\Big\|\mathcal
	{M}_{Z_i}\left[\begin{matrix}\Gamma_i \\\Pi_i\end{matrix}\right]-\Pi_iS\Big\|_F+\|\mathcal
	{M}_{C_i}\Pi_i-H\|_F.
\end{equation}
Overall, Problem \eqref{eq:opt} is convex and can be efficiently solved using off-the-shelf solvers. Let $(\Pi_i^*,\Gamma_i^*)$ denote any optimal solution of  \eqref{eq:opt}  and $(\Delta_{i1}^*,\Delta_{i2}^*)$ denote the resulting error of output regulation equations in \eqref{eq:regu:noise}.


The following result provides upper bounds on the unknown regulator equation errors $\Delta_{i1}^*$ and $\Delta_{i2}^*$, which serves as a solid basis for the subsequent analysis and design of data-driven output synchronization. For simplicity, we omit the subscript $F$ in the sequel, using $\|\cdot\|$ for $\|\cdot\|_F$.

\begin{lemma}[Bounded regulator equation errors]\label{lem:bound}
	Consider the MAS \eqref{mas}-\eqref{leader} with the relaxed output regulator equations \eqref{eq:regu:noise}. Suppose that Assumptions \ref{graph}-\ref{rank} hold. For any $(A_i,B_i,C_i)\in \mathcal{M}_{i}$ and $i\in \{1,2,\ldots,N\}$, there exist two bounded matrix polytopes $\mathcal{M}_{\Delta_{i1}}$ and $\mathcal{M}_{\Delta_{i2}}$ such that the regulator equation errors $\Delta_{i1}\in \mathcal{M}_{\Delta_{i1}}$ and $\Delta_{i2}\in\mathcal{M}_{\Delta_{i2}}$.
\end{lemma}

\begin{proof}
	It follows from the data-based polytopic representation in Lemma \ref{representation} that the true system matrices $(\bar{A}_i,\bar{B}_i,\bar{C}_i)\in \mathcal{M}_i$ can be expressed by a unique set of $\bar{\beta}_{W,i}^{(k)}$ and $\bar{\beta}_{V,i}^{(k)}$, respectively, as follows:
\begin{equation}\label{eq:rep:true}
	\begin{split}
	[\bar{B}_i\;\, \bar{A}_i]&=\Big(X_{i+}-\sum_{k=1}^{\gamma_{w_i}\rho}\bar{\beta}_{W,i}^{(k)}\hat{W}_{i}^{(k)}\Big)\left[\begin{matrix}U_i \\X_i\end{matrix}\right]^{\dagger}\\
	\bar{C}_i&=\Big(Y_{i}-\sum_{k=1}^{\gamma_{v_i}\rho}\bar{\beta}_{V,i}^{(k)}\hat{V}_{i}^{(k)}\Big)X_i^{\dagger}.
	\end{split}
\end{equation}

When $\bar{\beta}_{W,i}^{(k)}$ and $\bar{\beta}_{V,i}^{(k)}$ are known, we denote the solution of \eqref{regulator} associated with $(\bar{A}_i,\bar{B}_i,\bar{C}_i)$ by $(\Pi_i^s, \Gamma_i^s)$, which also serves as a candidate solution for the optimization problem \eqref{eq:opt}. 
Nonetheless, it may not necessarily yield the optimal objective value.
By utilizing \eqref{regulator}, \eqref{M_AB}, and \eqref{eq:rep:true}, it can be shown that the first term of \eqref{eq:opt} satisfies
	\begin{equation}
	\label{eq:term1}
	\begin{split}
	&
	\quad ~\Big\|{M}_{Z_i}\left[\begin{matrix}\Gamma_i^s \\\Pi_i^s\end{matrix}\right]-\Pi_i^sS\Big\|\\
	&=\Big\|(X_{i+}-\mathcal{M}_{W_i})\left[\begin{matrix}U_i \\X_i\end{matrix}\right]^{\dagger}\left[\begin{matrix}\Gamma_i^s \\\Pi_i^s\end{matrix}\right]-[\bar{B}_i\;\, \bar{A}_i]\left[\begin{matrix}\Gamma_i^s \\\Pi_i^s\end{matrix}\right]\Big\|\\
&=\Big\|\sum_{k=1}^{\gamma_{w_i}\rho}
\big[(\bar{\beta}_{W,i}^{(k)}-{\beta}_{W,i}^{(k)})\hat{W}_{i}^{(k)}\big]
\left[\begin{matrix}U_i \\X_i\end{matrix}\right]^{\dagger}\left[\begin{matrix}\Gamma_i^s \\\Pi_i^s\end{matrix}\right] \Big \|\\
		& \leq 2 \gamma_{w_i}\rho\bar{W}_{i}\Big\|\left[\begin{matrix}U_i \\X_i\end{matrix}\right]^{\dagger}\left[\begin{matrix}\Gamma_i^s \\\Pi_i^s\end{matrix}\right] \Big \|
	\end{split}
\end{equation}
with $\bar{W}_{i}:=\max_{k\in[1,\gamma_{w_i}\rho]}\|\hat{W}_{i}^{(k)}\|$.

Similarly, for the second term of \eqref{eq:opt}, we have
\begin{equation}\label{eq:term2}
	\begin{split}
		\mathcal
		\|{M}_{C_i}\Pi_i^s-H\| &=\Big\|\sum_{k=1}^{\gamma_{v_i}\rho}\big[(\bar{\beta}_{V,i}^{(k)}-{\beta}_{V,i}^{(k)})\hat{V}_{i}^{(k)}\big] X_i^{\dagger}\Pi_i^s \Big \|\\
		&\leq 2 \gamma_{v_i}\rho\bar{V}_{i} \|X_i^{\dagger}\Pi_i^s\|
	\end{split}
\end{equation}
with $\bar{V}_{i}:=\max_{k\in[1,\gamma_{v_i}\rho]}\|\hat{V}_{i}^{(k)}\|$.

By substituting the optimal solution $(\Pi_i^*, \Gamma_i^*)$ of problem \eqref{eq:opt} into \eqref{eq:regu:noise}, it follows from \eqref{eq:term1} and \eqref{eq:term2} that the regulator equation errors $\Delta_{i1}^*$ and $\Delta_{i2}^*$ adhere to
\begin{subequations}\label{eq:poly:delta}
\begin{align}
	\Delta_{i1}^*&\in\! \sum_{k=1}^{\gamma_{w_i}\rho}\big[(\bar{\beta}_{W,i}^{(k)}-{\beta}_{W,i}^{(k)})\hat{W}_{i}^{(k)}\big]\left[\begin{matrix}U_i \\X_i\end{matrix}\right]^{\dagger}\left[\begin{matrix}\Gamma_i^s \\\Pi_i^s\end{matrix}\right]\! \triangleq\! \mathcal{M}_{\Delta_{i1}}\\
	\Delta_{i2}^*&\in\! \sum_{k=1}^{\gamma_{v_i}\rho}\big[(\bar{\beta}_{V,i}^{(k)}-{\beta}_{V,i}^{(k)})\hat{V}_{i}^{(k)}\big] X_i^{\dagger}\Pi_i^s\triangleq \mathcal{M}_{\Delta_{i2}}
\end{align}
\end{subequations}
where $\mathcal{M}_{\Delta{i1}}$ and $\mathcal{M}_{\Delta_{i2}}$ are bounded polytopes, implying the boundedness of the regulator equation errors $\Delta_{i1}^*$ and $\Delta_{i2}^*$.
\end{proof}

It is worth emphasizing that \eqref{eq:opt} provides a data-driven method to compute the gain matrices $\Pi_i$ and $\Gamma_i$ independently of system matrices. Furthermore, Lemma \ref{lem:bound} ensures the practicability of the obtained $\Pi_i$ and $\Gamma_i$ for the subsequent analysis and design of data-driven output synchronization by constraining the regulator equation errors $\Delta_{i1}$ and $\Delta_{i2}$ within the bounded matrix polytopes $\mathcal{M}_{\Delta_{i1}}$ and $\mathcal{M}_{\Delta_{i2}}$.

\begin{remark}[Relationship between noise and output regulator equations]
\label{rem:regulator}
In fact, problem \eqref{eq:opt} addresses a variant of the output regulation problem that incorporates unknown noise. This implies that the solutions $\Pi_i$ and $\Gamma_i$ of \eqref{eq:opt} satisfy the relaxed equations \eqref{eq:regu:noise} when the noise $w_i(T)$ and $v_i(T)$ is not identically equal to zero.
According to \eqref{eq:poly:delta}, the size of the matrix polytopes $\mathcal{M}_{\Delta_{i1}}$ and $\mathcal{M}_{\Delta_{i2}}$ depends on the upper bound of noise, indicating that $\Delta_{i1}$ and $\Delta_{i2}$ increase with the noise levels $w_i(T)$ and $v_i(T)$. However, in the noise-free case where $w_i(T)=0$ and $v_i(T)=0$, the solution of \eqref{eq:opt} achieves zero cost (i.e., $\Delta_{i1}=0$ and $\Delta_{i2}=0$) and satisfies \eqref{regulator}.

\end{remark}





\subsection{Controller Design and Output Synchronization Analysis}

In the following, we focus on Problem \ref{problem1}, which involves learning stabilizing controllers $K_i$ for each follower and analyzing the output synchronization of the MAS \eqref{mas}-\eqref{leader} under the proposed distributed control protocol (\eqref{controller}, \eqref{observer}, and \eqref{eq:regu:noise}). To begin with, we reconstruct the dynamics of the virtual tracking error $\xi_i(t)$ as follows:
\begin{equation}
\label{eq:recon:aux}
	\begin{split}
	\xi_i(t+1)&=(\bar{A}_i+\bar{B}_iK_i)\xi_i(t)+\tilde{\delta}_i(t)
	\end{split}
\end{equation}
where we utilize \eqref{mas}, \eqref{controller}, \eqref{observer}, \eqref{eq:regu:noise}, and define $\tilde{\delta}_i(t) := \Delta_{i1}(x_0(t)+\delta_i(t)) + \Pi_i(1+d_i+g_i)^{-1}Fz_i(t)$.

We observe that $\xi_i(t)$ belongs to a well-defined polytope $\mathcal{P}_{\xi_i,t}$, i.e., $\xi_i(t)\in \mathcal{P}_{\xi_i,t}$ for $t\in \mathbb{N}$ and $i\in\{1,2,\ldots,N\}$. The next lemma provides the definition and boundedness guarantee of the polytope $\mathcal{P}_{\xi_i,t}$.

\begin{lemma}[Virtual tracking error polytope] 
\label{lem:auxerr:poly}
Under Assumptions~\ref{graph}-\ref{rank}, let $\mathcal{P}_{\xi_i,0} = \xi_i(0)$. At time $t\in \mathbb{N}$, the reachable set of the virtual tracking error is given by:
\begin{equation}
\label{eq:auxerr:poly}
	\mathcal{P}_{\xi_i,t}:=(\bar{A}_i+\bar{B}_iK_i)^t\xi_i(0)+\sum_{\ell=0}^{t-1}(\bar{A}_i+\bar{B}_iK_i)^\ell\mathcal{P}_{\tilde{\delta}_i,t-\ell-1}
\end{equation}
where $\mathcal{P}_{\tilde{\delta}_i,t}:=\Delta_{i1}^*\mathcal{P}_{x_0}+\Delta_{i1}^*\delta_i(t)+\Pi_i(1+d_i+g_i)^{-1}Fz_i(t)$, and $\mathcal{P}_{x_0}$ represents a well-defined polytope. Moreover, if there exist matrices $K_i$ and $F$ such that $\bar{A}_i+\bar{B}_iK_i$ and $I_N\otimes S- (I_N+D+G)^{-1}(\mathcal{L}+G)\otimes F$ are Schur stable, then $\mathcal{P}_{\xi_i,t}$ is a uniformly bounded set for any $t\in \mathbb{N}$ and $i\in \{1,2,\ldots,N\}$.

\end{lemma}

\begin{proof} 
	First, we invoke Assumption~\ref{pole} to guarantee the boundedness of the leader's state $x_0(t)$ for any $t\in \mathbb{N}$. By constructing a well-defined polytope $\mathcal{P}_{x_0}$, we ensure that $x_0(t)\in \mathcal{P}_{x_0}$ holds for all $t\in \mathbb{N}$. Moreover, considering the dynamics of the observer estimation error $\delta_i(t)$ in \eqref{eq:esti:entire}, we observe that if $I_N\otimes S- (I_N+D+G)^{-1}(\mathcal{L}+G)\otimes F$ is Schur stable, the observer convergence is assured, i.e., $\lim_{t\to \infty}\delta_i(t)= 0$. As a result, $\delta_i(t)$ and, consequently, $z_i(t)$ defined in \eqref{eq:auxerr:dyn}, are uniformly bounded. Building upon these findings, we deduce from \eqref{eq:recon:aux} that $\tilde{\delta}_i(t)$ forms a uniformly bounded sequence. Thus, there exists a bounded polytope $\mathcal{P}_{\tilde{\delta}_i,t}$ defined as $\mathcal{P}_{\tilde{\delta}_i,t}:=\Delta_{i1}^*\mathcal{P}_{x_0}+\Delta_{i1}^*\delta_i(t)+\Pi_i(1+d_i+g_i)^{-1}Fz_i(t)$ such that $\tilde{\delta}_i(t)\in\mathcal{P}_{\tilde{\delta}_i,t}$ holds for $t\in\mathbb{N}$.
	
Furthermore, due to the Schur stability of $\bar{A}_i+\bar{B}_iK_i$ and the boundedness of $\delta_i(0)$, there exists a uniformly bounded set, defined as in \eqref{eq:auxerr:poly}, denoted by $\mathcal{P}_{\xi_i,t}$. This ensures that $\xi_i(t) \in \mathcal{P}_{\xi_i,t}$ for any $t\in \mathbb{N}$ and $i\in \{1,2,\ldots,N\}$, thereby completing the proof.
\end{proof}

\begin{remark}[Design a stabilizing matrix $F$]
	\label{matrixF}
	We would like to emphasize that there exist techniques for computing a stabilizing matrix $F$ while ensuring the condition that $I_N\otimes S$ $- (I_N+D+G)^{-1}(\mathcal{L}+G)\otimes F$ is Schur stable. An effective approach is to solve discrete-time Riccati equations, as described in \cite{Kristian2013}, which provides a matrix $F$ satisfying the desired stability condition.
	
\end{remark}

At this stage, two challenges need to be addressed. First, the polytope $\mathcal{P}_{\xi_i,t}$ derived in Lemma~\ref{lem:auxerr:poly} cannot be directly applied in practical scenarios due to the unknown true system matrices $\bar{A}_i$ and $\bar{B}_i$. Second, it is crucial to ensure the stability of $\bar{A}_i+\bar{B}_iK_i$. To tackle the former, we aim to construct a conservative approximation $\bar{\mathcal{P}}_{\xi_i,t}$ of $\mathcal{P}_{\xi_i,t}$ such that $\mathcal{P}_{\xi_i,t}\subseteq \bar{\mathcal{P}}_{\xi_i,t}$. The latter issue can be addressed by designing a controller gain $K_i$ that stabilizes all $(A_i,B_i)$ in $\Sigma_{i}$.

 As previously mentioned, we proceed to construct an approximation $\bar{\mathcal{P}}_{\xi_i,t}$ of agent $i$. First, we observe from Lemma~\ref{lem:auxerr:poly} that
\begin{equation}
\begin{split}
	\mathcal{P}_{\xi_i,t}&=(\bar{A}_i+\bar{B}_iK_i)\mathcal{P}_{\xi_i,t-1}+\mathcal{P}_{\tilde{\delta}_i,t}.
\end{split}
\end{equation}
To this end, we define the matrix polytope as
\begin{equation}\label{eq:Mkz}
	\mathcal{M}_{Z_i}^K:=\mathcal{M}_{Z_i}\left[\begin{matrix}K_i \\I\end{matrix}\right].
\end{equation}
The approximation $\bar{\mathcal{P}}_{\xi_i,t}$ is obtained through the following lemma.

\begin{lemma}[Virtual tracking error polytope approximation] \label{lem:poly:approxi}
For $i\in\{1,2,\ldots,N\}$, let $\bar{\mathcal{P}}_{\xi_i,t}$ be defined as
	\begin{equation}
	\label{eq:lem5}
	\begin{split}
			\bar{\mathcal{P}}_{\xi_i,t}&:=\mathcal{M}_{Z_i}^{K}\bar{\mathcal{P}}_{\xi_i,t-1}+\mathcal{P}_{\tilde{\delta}_i,t}
	\end{split}
	\end{equation}
with $\bar{\mathcal{P}}_{\xi_i,0}=\mathcal{P}_{\xi_i,0}$. Then, it holds that $\mathcal{P}_{\xi_i,t}\subseteq \bar{\mathcal{P}}_{\xi_i,t}$ for any $t\in \mathbb{N}$.
\end{lemma}

\begin{proof}
	The validity of \eqref{eq:lem5} for $t = 0$ is straightforward. For a fixed $t>0$, leveraging the induction step, we have $\mathcal{P}_{\xi_i,t-1}\subseteq \bar{\mathcal{P}}_{\xi_i,t-1}$ and $\bar{A}_i+\bar{B}_iK_i=Z_i\left[\begin{matrix}K_i\\ I\end{matrix}\right]\in \mathcal{M}_{Z_i}\left[\begin{matrix}K_i\\ I\end{matrix}\right]$. Consequently, we obtain $(\bar{A}_i+\bar{B}_iK_i)\mathcal{P}_{\xi_i,t-1}\subseteq\mathcal{M}_{Z_i}^{K}\bar{\mathcal{P}}_{\xi_i,t-1}$, which completes the proof.

\end{proof}

Drawing on the aforementioned lemma, we can build on the result in \cite{russo2022tube} and establish the following lemma to ensure the stability of the virtual tracking error polytope $\bar{\mathcal{P}}_{\xi_i,t}$.

\begin{lemma}[Stability of the virtual tracking error polytope] 
\label{lem:stab}
Given Assumptions~\ref{graph}-\ref{rank} and an initial value $\xi_i(0)$ for agent $i\in \{1,2,\ldots,N\}$, there exists a polytope $\bar{\mathcal{P}}_i\subset\mathbb{R}^{n_i}$ satisfying the following properties: i) $\bar{\mathcal{P}}_{\xi_i,t}\subset\bar{\mathcal{P}}_i$ for all $t\in\mathbb{N}$; and ii) $\bar{\mathcal{P}}_i$ is an invariant set, i.e., for $\xi_i(t)\in \bar{\mathcal{P}}_i$, it holds that $\mathcal{M}_{Z_i}^{K}\xi_i(t)+\tilde{\delta}_i(t)\in \bar{\mathcal{P}}_i$ for all $\tilde{\delta}_i(t)\in \mathcal{P}_{\tilde{\delta}_i,t}$ and $t\in\mathbb{N}$.

\end{lemma}

\begin{proof}
The proof consists of two steps. In the first step, we compute the reachable set of $\bar{\mathcal{P}}_{\xi_i,t}$. In the second step, we ensure the stability of this reachable set by proving its invariance, which in turn guarantees the stability of $\bar{\mathcal{P}}_{\xi_i,t}$.

First, recalling Lemma~\ref{lem:auxerr:poly}, we refer to the bounded and compact set $\mathcal{P}_{\tilde{\delta}i,t}$ as the disturbance set. Building on this set, $\bar{\mathcal{P}}_{\xi_i,t}$, and $\mathcal{M}_{Z_i}^{K}$ defined in \eqref{eq:Mkz}, the reachable set for the virtual tracking error $\xi_i(t)$ at time $t$ can be formulated as
\begin{equation}\label{eq:Xi}
	\Xi_{i,t}=\left\{Q_i\Xi_{i,t-1}+\mathcal{P}_{\tilde{\delta}_i,t}: Q_i\in\mathcal{M}_{Z_i}^{K}\right\}
\end{equation}
with $\xi_i(0)=\Xi_{i,0}$. Hence, it follows from Lemma~\ref{lem:poly:approxi} that $\bar{\mathcal{P}}_{\xi_i,t}\subset\Xi_{i,t}$.

Next, we prove the stability of the reachable set $\Xi_{i,t}$.
Assumption \ref{ABC} states that the system is stabilizable, implying the existence of a positive definite symmetric matrix ${P}_i$ such that $Q_i^\top {P}_iQ-{P}_i\prec 0$ for all $Q_i\in \mathcal{M}_{Z_i}^K$.
Therefore, for sufficiently small $\mathcal{P}_{\tilde{\delta}_i,t}$, we have $\Xi_{i,t+1}\subseteq\Xi_{i,t}$ for $t\in\mathbb{N}$, ensuring the stability of the reachable trajectories $\Xi_{i,t}$.
Consequently, according to Lemma~\ref{lem:poly:approxi}, there exists an invariant set $\bar{\mathcal{P}}_i$ satisfying $\bar{\mathcal{P}}_i\supset\Xi_{i,t}\supset \bar{\mathcal{P}}_{\xi_i,t}$ for all $t\in\mathbb{N}$, with $\bar{\mathcal{P}}_i:=\Xi_{i,0}$.
The proof is complete.
\end{proof}

Lemma \ref{lem:stab} provides a stability guarantee for $\bar{\mathcal{P}}_{\xi_i,t}$, which also implies the stability of the virtual tracking error polytope ${\mathcal{P}}_{\xi_i,t}$ since $\mathcal{P}_{\xi_i,t}\subseteq \bar{\mathcal{P}}_{\xi_i,t}$. With this result in hand, we proceed to the problem of identifying a gain matrix $K_i$ for follower $i$. This gain matrix ensures that $A_i + B_iK_i$ is Schur stable for all $(A_i, B_i) \in \mathcal{M}_{Z_i}$. We derive a convex program, specifically a SDP, in the next theorem. This program aims to search for a stabilizing matrix $K_i$ based on  $(X_i,X_{i+},U_i,Y_i)$.  

\begin{thm}
\label{thm:k}

Consider the MAS \eqref{mas}-\eqref{leader} under the distributed data-driven feedback protocol \eqref{controller} and \eqref{observer} over graph $\bar{\mathcal{G}}$. Let Assumptions~\ref{graph}-\ref{rank} hold. Define $\Omega_i=X_{i+}-\sum_{k=1}^{\gamma_{w_i}\rho}\beta_{W,i}^{(k)}\hat{W}_{i}^{(k)}$.  Then, the following SDP is feasible, and the gain matrix $K_i=U_iM_i(X_iM_i)^{-1}$ with any $M_i\in\mathbb{R}^{\rho\times n_i}$ satisfying \eqref{sdp} renders $A_i+B_iK_i$ Schur stable for all $(A_i, B_i)\in \mathcal{M}_{Z_i}$
\begin{equation}
	\label{sdp}
	\begin{split}
		X_iM_i-\Omega_iM_i(X_iM_i)^{-1}(\Omega_iM_i)^\top&\succ0\\
		X_iM_i&\succ0.
	\end{split}
\end{equation}

\end{thm}

\begin{proof}
	Consider any matrix $K_i$ that makes $A_i+B_iK_i$ Schur stable. Based on the proof of Lemma \ref{lem:stab}, for all $(A_i, B_i)\in \mathcal{M}_{Z_i}$, there exists a matrix $P_i\succ 0$ such that
\begin{equation}
	\label{Lyap}
	(A_i+B_iK_i)^\top P_i(A_i+B_iK_i)-P_i\prec 0.
\end{equation}

According to Assumption~\ref{rank}, any vector of length $n_i+m_i$ can be expressed as a linear combination of the data matrix. That is, there exists a matrix $G_{i}\in\mathbb{R}^{\rho\times n_i}$ satisfying
\begin{equation}
	\label{XG=I}
	\left[\begin{matrix}K_i \\I\end{matrix}\right]=\left[\begin{matrix}U_i \\X_i\end{matrix}\right]G_i
\end{equation}
which implies
\begin{equation}
	\label{AB3}
	\begin{split}
		A_i+B_iK_i&=Z_i\left[\begin{matrix}U_i \\X_i\end{matrix}\right]G_i\\
		&=\Big(X_{i+}-\sum_{k=1}^{\gamma_{w_i}\rho}\beta_{W,i}^{(k)}\hat{W}_{i}^{(k)}\Big)G_i
	\end{split}
\end{equation}
with the second identity following from \eqref{AB2}.
Substituting \eqref{AB3} into \eqref{Lyap}, we have
\begin{equation}
\label{eq:xpx}
	\begin{split}
		\Big[\Big(X_{i+}&-\sum_{k=1}^{\gamma_{w_i}\rho}\beta_{W,i}^{(k)}\hat{W}_{i}^{(k)}\Big)G_i\Big]^\top P_i\Big[\Big(X_{i+}\\
	&-\sum_{k=1}^{\gamma_{w_i}\rho}\beta_{W,i}^{(k)}\hat{W}_{i}^{(k)}\Big)G_i\Big]-P_i\prec0.
	\end{split}
\end{equation}

Let $G_i=M_iP_i$ and $\Omega_i=X_{i+}-\sum_{k=1}^{\gamma_{w_i}\rho}\beta_{W,i}^{(k)}\hat{W}_{i}^{(k)}$. 
By a Schur complement argument, condition \eqref{eq:xpx} implies
\begin{equation}
	\label{schur}
		P_i^{-1}-\Omega_iM_iP_i(\Omega_iM_i)^{\top}\succ0.
\end{equation}
Recall from \eqref{XG=I} that $X_iG_i=I$ and $K_i=U_iG_i$. By exploiting $P_i=(X_iM_i)^{-1}$, seeking stability is equivalent to finding a matrix $M_i$ such that
\begin{equation*}
		X_iM_i-\Omega_iM_i(X_iM_i)^{-1}(\Omega_iM_i)^{\top}\succ0
\end{equation*}
with the choice $K_i=U_iM_i(X_iM_i)^{-1}$.
Then, it can be concluded that $A_i+B_iK_i$ is stable as $P_i$ is symmetric and $P_i\succ 0$ for all $(A_i,\, B_i)\in \mathcal{M}_{Z_i}$.
\end{proof}


Before presenting our stability result, we define the tracking error $e_i(t)$ of agent $i$ as $e_i(t):=y_i(t)-y_0(t)$. Our objective is to analyze the convergence of the tracking error, ensuring that the induced closed-loop system exhibits desired stability properties. Leveraging previous results, including the optimization problem~\eqref{eq:opt}, Lemmas \ref{lem:bound}-\ref{lem:stab}, and Theorem~\ref{thm:k}, we present our distributed data-driven output synchronization solution for unknown heterogeneous MASs \eqref{mas}-\eqref{leader} in Algorithm~\ref{Alg}, accompanied by the stability guarantees detailed below.

\begin{algorithm}[!htb]
	\caption{Distributed data-driven output synchronization}
	\label{Alg}
	\begin{algorithmic}[1]
		\State \textbf{Input:} desired lifespan of the MAS $\mathcal{T}$; initial states of the leader $x_0(0)\in \mathbb{R}^{n_0}$ and follower $x_i(0) \in \mathbb{R}^{n_i}$; noise polytopes $\mathcal{P}_{w_i}$ and $\mathcal{P}_{v_i}$; regulator equation error polytopes $\mathcal{M}_{\Delta{i1}}$ and $\mathcal{M}_{\Delta{i2}}$;
		state-input-output data $\mathbb{D}_i$ for each $ i\in \{1,2,\ldots,N\}$.
		\State \textbf{Construct} matrix polytopes $\mathcal{M}_{W_{i}}$ and $\mathcal{M}_{V{i}}$ via \eqref{M_poly}.
		\State \textbf{Compute} data-based polytopic representation $\mathcal{M}_{i}$ using Lemma \ref{representation}.
		\State \textbf{Solve} Problem \eqref{eq:opt} to obtain feasible solutions $\Pi_i$ and $\Gamma_i$.
		\State \textbf{Design} the controller gain matrix $K_i$ and observer gain matrix $F$ based on Theorem \ref{thm:k}.
		\While {$t<\mathcal{T}$}
		\For{$i=1,2,\ldots, N$}
		\State{Broadcast $\eta_i(t)$ and $x_i(t)$ to agent $j\in \mathcal{N}_i$;}
		\State {Compute $z_i(t)$ from \eqref{eq:auxerr:dyn} using the updated state $\eta_j(t)$ and $x_j(t)$, $j\in \mathcal{N}_i$;}
		\State{Update the feedback control protocol \eqref{controller}, \eqref{observer}, and the dynamics \eqref{mas} of agent $i$.}
		\EndFor
		\EndWhile
	\end{algorithmic}
\end{algorithm}

\begin{thm}
\label{thm:uub}
Consider the MAS described by \eqref{mas}-\eqref{leader} and the graph $\bar{\mathcal{G}}$. Suppose Assumptions~\ref{graph}-\ref{rank} hold. 
Denoting the optimal solution of Problem \eqref{eq:opt} by $(\Pi_i^*,\Gamma_i^*)$, the tracking errors $e_i(t)$ are ultimately uniformly bounded (UUB) under the distributed data-driven feedback protocol \eqref{controller} and \eqref{observer} for any initial state and all $i\in \{1,2,\ldots,N\}$, if the following two conditions are satisfied:
 \begin{enumerate}
 	\item The controller gain $K_i$ is designed as in Theorem~\ref{thm:k}.
 	\item Choose matrix $F$ such that $I_N\otimes S- (I_N+D+G)^{-1}(\mathcal{L}+G)\otimes F$ is Schur stable.
 \end{enumerate}
\end{thm}

\begin{proof}
It can be observed that the solutions $\Pi_i$ and $\Gamma_i$ to the optimization problem \eqref{eq:opt} satisfy \eqref{eq:regu:noise} for all $(A_i,B_i,C_i)\in \mathcal{M}_{i}$.
Considering the true system matrices $(\bar{A}_i,\bar{B}_i,\bar{C}_i)\in \mathcal{M}_{i}$, we deduce from \eqref{mas}, \eqref{leader}, and \eqref{eq:regu:noise} that the tracking error $e_i(t)$ can be expressed as
\begin{equation}
	\label{eq:E1}
		e_i(t)=\bar{C}_ix_i(t)-(\bar{C}_i\Pi_i^*-\Delta_{i2}^*)x_0(t).
\end{equation}

Moreover, based on \eqref{eq:esti:entire}, if we select a stabilizing matrix $F$ that satisfies the condition $I_N\otimes S- (I_N+D+G)^{-1}(\mathcal{L}+G)\otimes F$ being Schur stable, the observer's state asymptotically converges to the leader's state, i.e., $\lim_{t\to \infty}\delta_i(t)=\eta_i(t)-x_0(t)= 0$. Therefore, as $t\rightarrow \infty$, we have
\begin{equation}\label{eq:track:infty}
	\lim_{t\to \infty}e_i(t)=\lim_{t\to \infty}\bar{C}_i\xi_i(t)+\Delta_{i2}^*x_0(t),
\end{equation}
where the fact that $\xi_i(t):=x_i(t)-\Pi_i^*\eta_i(t)$ has been used.

Furthermore, we can recursively obtain from \eqref{eq:lem5} that for $t\in\mathbb{N}$, the polytope $\bar{\mathcal{P}}_{\xi_i,t}$ satisfies
\begin{align}
	\bar{\mathcal{P}}_{\xi_i,t}=(\mathcal{M}_{Z_i}^{K})^t\bar{\mathcal{P}}_{\xi_i,0}+\sum_{\ell=0}^{t-1}(\mathcal{M}_{Z_i}^{K})^\ell\mathcal{P}_{\tilde{\delta}_i,t-\ell-1}
\end{align}

As $\lim_{t\to \infty}\delta_i(t)= 0$, it follows that $\lim_{t\to \infty}z_i(t)=0$. Consequently, as $t\rightarrow \infty$, the set $\bar{\mathcal{P}}_{\xi_i,t}$ converges to the following set:
\begin{equation}\label{eq:barP}
	(\mathcal{M}_{Z_i}^{K})^t\bar{\mathcal{P}}_{\xi_i,0}+\sum_{\ell=0}^{t-1}(\mathcal{M}_{Z_i}^{K})^\ell\Delta_{i1}^*\mathcal{P}_{x_0}.
\end{equation}

Based on Lemmas~\ref{representation}-\ref{lem:stab}, it can be concluded that matrices $\bar{C}_i$, $\Delta_{i1}$, and $\Delta_{i2}$, as well as the sequences $\xi_i(t)$ and $x_0(t)$, can be bounded and constrained within compact polytopes, respectively. Specifically, we have $\bar{C}_i\in\mathcal{M}_{C_i}$, $\Delta_{i1}^*\in\mathcal{M}_ {\Delta_{i1}}$, $\Delta_{i2}^*\in\mathcal{M}_{\Delta_{i2}}$, $\xi_i(t)\in \mathcal{P}_{\xi_i,t}\subseteq \bar{\mathcal{P}}_{\xi_i,t}$, and $x_0(t)\in\mathcal{P}_{x_0}$ for all $t\in \mathbb{N}$.
By combining \eqref{eq:track:infty} and \eqref{eq:barP}, we can compute the reachable set of the tracking error $e_i(t)$ at time $t$ (as $t\rightarrow \infty$) as follows:
\begin{align}
	\mathcal{P}_{e_i,t}&:=\mathcal{M}_{C_i}(\mathcal{M}_{Z_i}^{K})^t\bar{\mathcal{P}}_{\xi_i,0}\nonumber \\
	&~\quad +\big[\mathcal{M}_{C_i}\sum_{\ell=0}^{t-1}(\mathcal{M}_{Z_i}^{K})^\ell\mathcal{M}_{\Delta_{i1}}+\mathcal{M}_	{\Delta_{i2}}\big]\mathcal{P}_{x_0}\label{eq:track:poly}.
\end{align}

According to Theorem~\ref{thm:k}, the controller gain matrix $K_i=U_iM_i(X_iM_i)^{-1}$ renders ${A}_i+{B}_iK_i$ Schur stable for all $(A_i,\, B_i)\in \mathcal{M}_{Z_i}$, which implies the existence of $P_i:=(X_iM_i)^{-1}$ such that any $Q_i\in \mathcal{M}_{Z_i}^K$ is also Schur stable.
This means that there exist constants $\mu>0$ and  $\beta=\max_{Q_i\in \mathcal{M}_{Z_i}^K}\lambda_{\min}(Q_i)\in(0,1)$ such that $(\mathcal{M}_{Z_i}^K)^t\bar{\mathcal{P}}_{\xi_i,0}\subseteq\mu\beta^t\bar{\mathcal{P}}_{\xi_i,0}$ holds.
Consequently, \eqref{eq:track:poly} obeys
\begin{equation}\label{eq:track:poly2}
	\begin{split}
		\mathcal{P}_{e_i,t}&\subseteq \mu\beta^t\mathcal{M}_{C_i}\bar{\mathcal{P}}_{\xi_i,0}\\
		&\quad+\Big[\mu\mathcal{M}_{C_i}\sum_{\ell=0}^{t-1}(\mathcal{M}_{Z_i}^K)^\ell\mathcal{M}_{\Delta_{i1}}+\mathcal{M}_{\Delta_{i2}}\Big]\mathcal{P}_{x_0}.
	\end{split}
\end{equation}
Hence, as $t\rightarrow \infty$, $e_i(t)$ converges to an adjustable bounded set as \eqref{eq:track:poly2}.

Therefore, it can be concluded that the tracking error $e_i(t)$ is UUB for all $i\in\{1,2,\ldots,N\}$ when considering the proposed distributed data-driven feedback protocol given by \eqref{controller} and \eqref{observer}. This property enables the outputs of all followers to approximately synchronize with the output of the leader.

\end{proof}

\begin{remark}[$\Delta$-optimal output synchronization]
	The proposed distributed data-driven feedback control protocol \eqref{controller}--\eqref{observer} achieves $\Delta$-optimal output synchronization for the leader-following MAS \eqref{mas}--\eqref{leader} with unknown system matrices. This result, presented in Theorem~\ref{thm:uub}, addresses Problem \ref{problem1} effectively. Specifically, for any $(A_i,B_i,C_i)\in \Sigma_i$ with $i\in\{1,2,\ldots,N\}$, the tracking error $e_i(t)$ converges to a bounded and compact reachable set as given in \eqref{eq:track:poly2}. The size of this set is influenced by the size of noise polytopes $\mathcal{M}_{\Delta_{i1}}$ and $\mathcal{M}_{\Delta_{i2}}$, as defined in \eqref{eq:poly:delta}. Notably, the magnitude of tracking error $e_i(t)$ is positively correlated with the vertices of noise polytopes $\mathcal{P}_{w_i}$ and $\mathcal{P}_{v_i}$, corresponding to the noise levels $w_i(T)$ and $v_i(T)$, respectively. 
	It is interesting that when $w_i(T)=0$ and $v_i(T)=0$, the regulator equation errors $\Delta_{i1}=0$ and $\Delta_{i2}=0$, resulting in the virtual tracking error $\xi_i(t)$ asymptotically converge to zero, as shown in Lemma~\ref{lem:stab}. In this case, the tracking error $e_i(t)$ asymptotically converges to zero too.

\end{remark}

Based on the above remark, we can derive the following corollary for the noise-free case.
	
\begin{corollary}
	Consider the leader-following MAS \eqref{masnoisy} with $w_i(T)=0$ and $v_i(T)=0$. Under the same conditions as in Theorem~\ref{thm:uub}, the tracking error $e_i(t)$ asymptotically converges to zero for all $(A_i,B_i,C_i)\in \Sigma_i$ and $i\in \{1,2,\ldots,N\}$.
\end{corollary}

\begin{remark}[Solvability]
	It is important to highlight that the solution to Problem \eqref{eq:opt} and the SDP \eqref{sdp} only needs to consider the vertices of the matrix polytopes $\mathcal{M}_{Z_i}$ and $\mathcal{M}_{C_i}$. This is due to the convexity of the polytope, which implies that any point within the polytope can be expressed as a convex combination of its vertices. Therefore, if the problem is addressed at all vertices of the polytope, it covers the entire polytope, including its interior.  
\end{remark}

\begin{remark}[Comparison]
	Several existing approaches have explored data-driven output synchronization for unknown MASs, e.g., the behavioral approach in \cite{Jiao2021} and the model-free RL-based approach in \cite{Kiumarsi2017,Jiang2023rl,Chen2023}.
	In comparison to these existing works, our approach exhibits the following key distinctions.
	Firstly, the proposed data-driven method is robust to unknown noisy data, eliminating the requirement for accurately measurable noise as seen in \cite{Jiao2021}.
	Secondly, in contrast to a large amount training data required for \cite{Jiang2023rl,Chen2023}, our approach achieves output synchronization only using limited data, i,e., as long as Assumption~\ref{rank} is satisfied.
	Furthermore, we propose a static data-driven design method using historical data, while ensuring system stability, which circumvents real-time iterative computations using online data, as in \cite{Jiang2023rl,Chen2023}.  
	In this sense,  the proposed method significantly reduces the computational burden, making it more efficient and suitable for real-world implementation.
\end{remark}


\section{Numerical Examples}
\label{section4}
In this section, we present a numerical example to demonstrate the effectiveness of the proposed data-driven method. We consider a discrete-time heterogeneous MAS consisting of seven agents, including one leader and six followers. The dynamics of the leader are described by \eqref{leader}, where
$$S=\left[\begin{matrix}0&1 \\-1 &0\end{matrix}\right], \quad H=\left[\begin{matrix}1&0 \end{matrix}\right].$$
The true dynamics of the six followers are given by \eqref{mas} with
\begin{align*}
	&\bar{A}_1=2,\bar{B}_1=3,\bar{C}_1=1\\
	&\bar{A}_2=\left[\begin{matrix}0&1 \\1 &-1\end{matrix}\right],\bar{B}_2=\left[\begin{matrix}0\\1\end{matrix}\right],\bar{C}_2=\left[\begin{matrix}1&1 \end{matrix}\right]\\
	&\bar{A}_3=\left[\begin{matrix}0&1 \\1 &-2\end{matrix}\right],\bar{B}_3=\left[\begin{matrix}1\\0\end{matrix}\right],\bar{C}_3=\left[\begin{matrix}0&1 \end{matrix}\right]\\
	&\bar{A}_4=\left[\begin{matrix}0&1 \\-1 &-3\end{matrix}\right], \bar{B}_4=\left[\begin{matrix}1\\1\end{matrix}\right], \bar{C}_4=\left[\begin{matrix}-1&1 \end{matrix}\right]\\
	&\bar{A}_5=\left[\begin{matrix}0&1&0 \\0 &0&1\\0&0&-4\end{matrix}\right],\bar{B}_5=\left[\begin{matrix}0\\0\\4\end{matrix}\right],\bar{C}_5=\left[\begin{matrix}1&0&0 \end{matrix}\right]\\
	&\bar{A}_6=\left[\begin{matrix}0&1&0 \\0 &0&1\\0&0&-5\end{matrix}\right],\bar{B}_6=\left[\begin{matrix}0\\0\\5\end{matrix}\right],\bar{C}_6=\left[\begin{matrix}2&0&0 \end{matrix}\right].
\end{align*}
The network topology of these agents is shown in Fig.~\ref{fig:graph}, which represents the communication topology $\bar{\mathcal{G}}$ between agents.
\begin{figure}[t]
	\centering
	\includegraphics[scale=0.65]{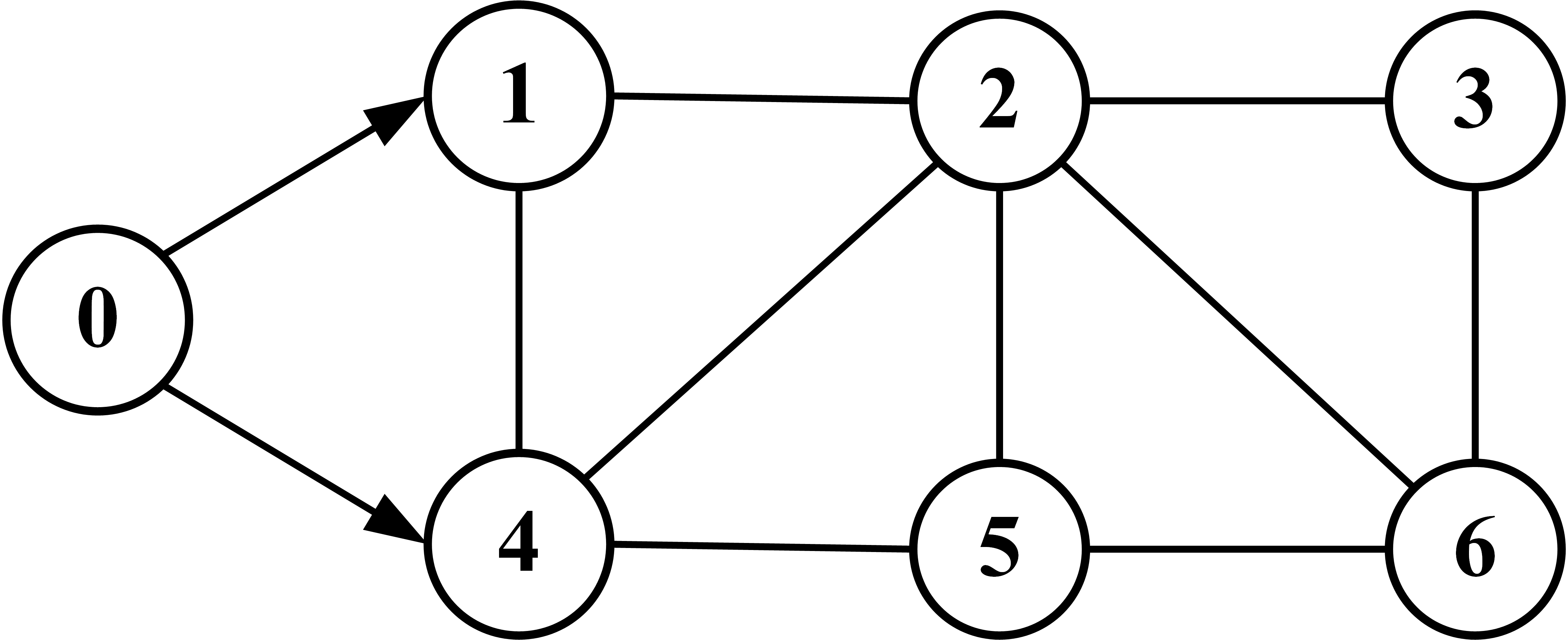}
\caption{The communication topology $\bar{\mathcal{G}}$ between agents.}
	\label{fig:graph}
\end{figure}

In the data-driven setting, the true system matrices $(\bar{A}_i,\bar{B}_i,\bar{C}_i)$ are assumed unknown. To collect data for each follower, we run the open-loop system \eqref{masnoisy} offline and gather a set of noisy data $\mathbb{D}_i$ with a length of $\rho=20$. The inputs are uniformly distributed within a polytope $\mathcal{P}_{u_i}$ defined as
 $$\mathcal{P}_{u_i}=\Big\{u_i\Big|u_i=\sum_{k=1}^{\gamma_{u_i}}\beta^{(k)}_{u,i}\hat{u}_{i}^{(k)},\beta^{(k)}_{u,i}\geq0, \sum_{k=1}^{\gamma_{u_i}}\beta^{(k)}_{u,i}=1 \Big\}$$
 where $\gamma_{u_i}=2$ and the two vertices are $\hat{u}_{i}^{(1)}=1$ and $\hat{u}_{i}^{(2)}=-1$ for $i\in\{1,2,\ldots,6\}$. Furthermore, the random process noise $w_i(T)$ is sampled from the polytope $\mathcal{P}_{w_i}$ with four vertices:
 $\hat{w}_i^{(1)}=\bar{w}_i*[1\, 1]^\top$,
$\hat{w}_i^{(2)}=\bar{w}_i*[1\, -1]^\top$,
$\hat{w}_i^{(3)}=\bar{w}_i*[-1\, 1]^\top$, and 
$\hat{w}_i^{(4)}=\bar{w}_i*[-1\, -1]^\top$.
 Similarly, the measurement noise $v_i(T)$ is bounded by $\mathcal{P}_{v_i}$ with two vertices: $\hat{v}_i^{(1)}=\bar{v}_i$ and 
$\hat{v}_i^{(2)}=-\bar{v}_i$. Here, select $\bar{w}_i=\bar{v}_i=0.01$.
 
To solve the output regulation equations \eqref{eq:regu:noise}, we need to compute the optimization problem \eqref{eq:opt} to determine the values of $(\Pi_i^*,\Gamma_i^*)$.
\begin{equation*}
\begin{split}
&\Pi_1^*=\left[\begin{matrix}1.0000&-0.0000\end{matrix}\right],\quad\,\Gamma_1^*=\left[\begin{matrix}-0.6672& 0.3336\end{matrix}\right],\\
	&\Pi_2^*=\left[\begin{matrix}0.5013&-0.4995 \\0.4987 &0.4995\end{matrix}\right],\quad\Gamma_2^*=\left[\begin{matrix}-0.5014&1.4994\end{matrix}\right],\\
	&\Pi_3^*=\left[\begin{matrix}1.9914&0.9954 \\0.9952 &-0.0020\end{matrix}\right],\quad\,
\Gamma_3^*=\left[\begin{matrix}-1.9907&1.9898\end{matrix}\right],\\
&\Pi_4^*=\left[\begin{matrix}-0.7995&-0.1998 \\0.1996 &-0.2002\end{matrix}\right],\;\,\Gamma_4=\left[\begin{matrix}-0.0002&-0.6007\end{matrix}\right],\\
&\Pi_5^*=\left[\begin{matrix}0.9999&0.0008\\ -0.0010&1.0010\\-1.0001&-0.0011\end{matrix}\right],\;\Gamma_5^*=\left[\begin{matrix}-1.0002  & -0.2515\end{matrix}\right]\\
&\Pi_6^*=\left[\begin{matrix}0.5004&0.0001\\-0.0004 & 0.5004\\-0.4996&-0.0004\end{matrix}\right],\;\Gamma_6^*=\left[\begin{matrix}-0.4991 & -0.1006\end{matrix}\right].\end{split}	
\end{equation*}

Then, by solving the SDP in Theorem \ref{thm:k}, we obtained the stabilizing controller gain $K_i$ for each follower $i$ in \eqref{controller}. The specific values of $K_i$ are as follows:
\begin{equation*}
\begin{split}
&K_1=-16.6608,\quad K_2=[-7.0686\;-1.8060],\\
&K_3=[-13.3323\; -10.3354],K_4=[-13.5107\; -15.1270],\\
&K_5=[-0.5139 \;-0.9054\;-0.5120],\\
&K_6=[-0.7749 \; -1.3071 \; -0.5160].
\end{split}	
\end{equation*}


\subsubsection {Comparison with the model-based approach}
The simulation of the heterogeneous MAS was carried out using the feedback control protocol \eqref{controller} and \eqref{observer}. The initial states of the leader, followers, and observers were randomly selected. The tracking errors $e_i$ for $i\in\{1,2,\ldots,N\}$ under the data-driven control (according to Theorem~\ref{thm:uub}) and model-based control (according to Lemma~\ref{lem1}) are shown in Fig.~\ref{fig:error}, respectively.
%
%

\begin{figure}[t]
	\centering
	\includegraphics[scale=0.42]{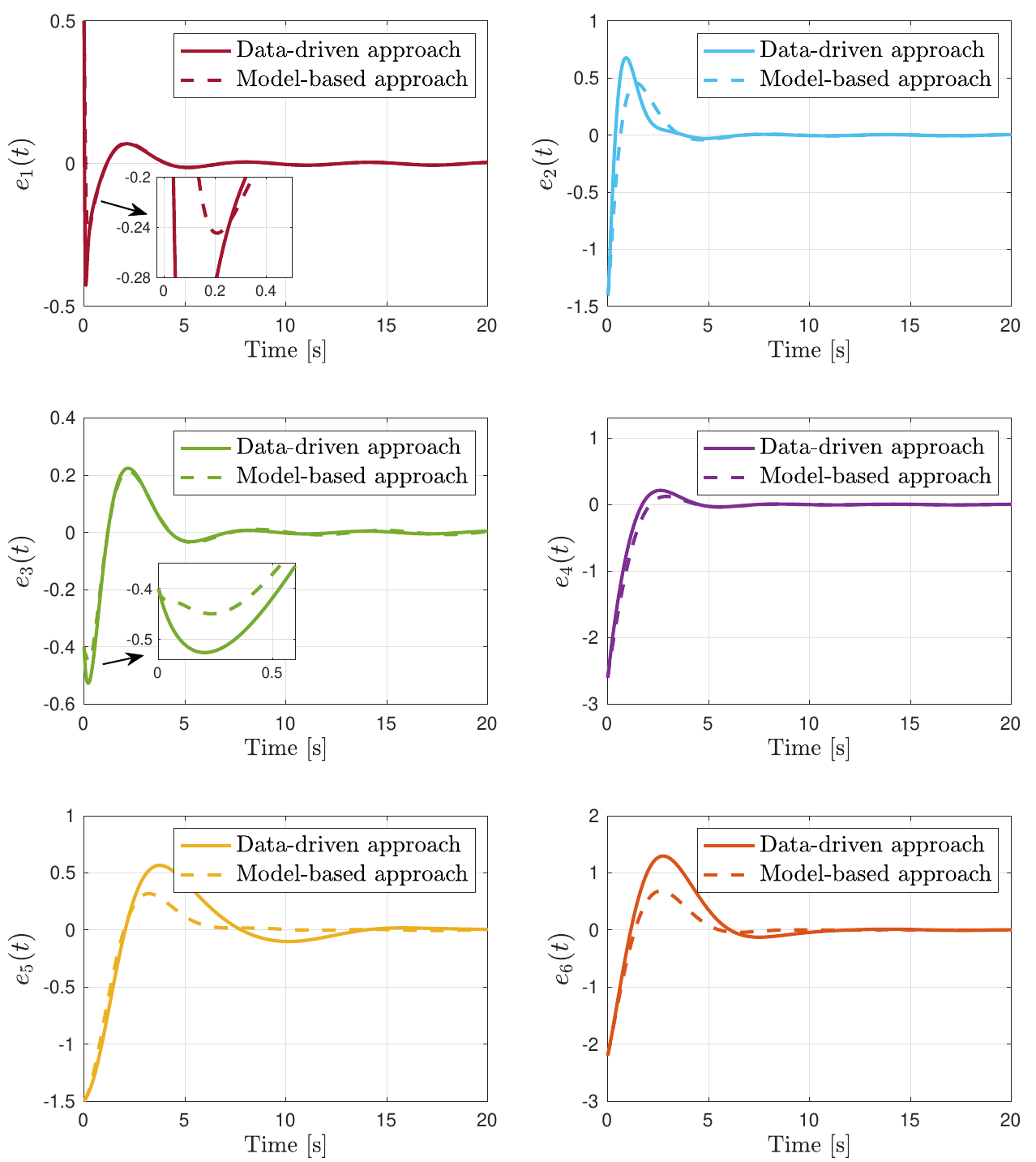}
	\caption{Tracking error of each follower under data-driven control and model-based control.}
	\label{fig:error}
\end{figure}

From this figure, it is evident that output synchronization is achieved under both control paradigms.
The results demonstrate that the proposed data-driven polytopic method achieves comparable performance to the model-based approach and exhibits excellent robustness to noisy and limited data, highlighting the effectiveness of the data-based polytopic controller. Additionally, the absence of model information during implementation further emphasizes the superiority of the data-driven method.

\subsubsection {Comparison of different noise levels}
The proposed data-driven approach was tested under different disturbance levels to investigate
their effect on system performance and tracking errors.
First, let us define the error between the solutions $(\Pi_i^s,\Gamma_i^s)$ and $(\Pi_i^*,\Gamma_i^*)$ of the
output regulation equation computed by \eqref{regulator} and \eqref{eq:opt}, respectively, as follows
\begin{equation*}
	\phi_1=\max_{i\in\{1,2,\ldots,6\}}\|\Pi_i^s-\Pi_i^*\|_2,\quad \phi_2=\max_{i\in\{1,2,\ldots,6\}}\|\Gamma_i^s-\Gamma_i^*\|_2.
\end{equation*}

Table \ref{table1} tabulates $\phi_1$ and $\phi_2$ for five different noise levels $\bar{w}_i$, $\bar{v}_i$. 
It can be observed that a larger noise level results in a lower accuracy of $(\Pi_i^*,\Gamma_i^*)$ relative to the exact solution $(\Pi_i^s,\Gamma_i^s)$.
\begin{table}[!htb]
	\centering
	\caption{The error between the regulation equation solutions computed by \eqref{regulator} and \eqref{eq:opt} under different noise levels}
	\label{table1}
	\resizebox{\linewidth}{!}{
		\begin{tabular}{cccccccc}
			\toprule 
			{Noise level}  & 0.001& 0.005&0.01&0.05& 0.1                  \\ \midrule 
			$\phi_1$    & $8.5149\times10^{-4}$  & 0.0051 &0.0110&0.0587&0.1463\\
			$\phi_2$  & $5.8310\times10^{-4}$  & 0.0040  & 0.0138 &0.0876&0.1748  \\	
			\bottomrule 
	\end{tabular}	}
\end{table}

Furthermore, taking follower $6$ as an example, Fig.~\ref{fig:bound} displays the tracking error $e_6(t)$ (the orange solid line) and the bounds of $\mathcal{P}_{e_6,t}$ (the blue dashed line) under three different noise levels. 
As expected, the tracking error $e_6(t)$ remains within $\mathcal{P}_{e_6,t}$ at each time step.
Moreover, the error bound \eqref{eq:track:poly2} is on the order (precisely, several times) of the noise size.
 This can be attributed to the fact that 
the escalating uncertainties caused by noise contribute to the expansion of noise polytopes $\mathcal{M}_{\Delta_{i1}}$ and $\mathcal{M}_{\Delta_{i2}}$. 
As a result, the polytope $\mathcal{M}_i$ of allowable system matrices becomes larger, significantly augmenting the conservatism of the obtained data-driven control solutions.


\begin{figure}[tb]
	\centering
	\subfloat[]{\includegraphics[scale=0.25]{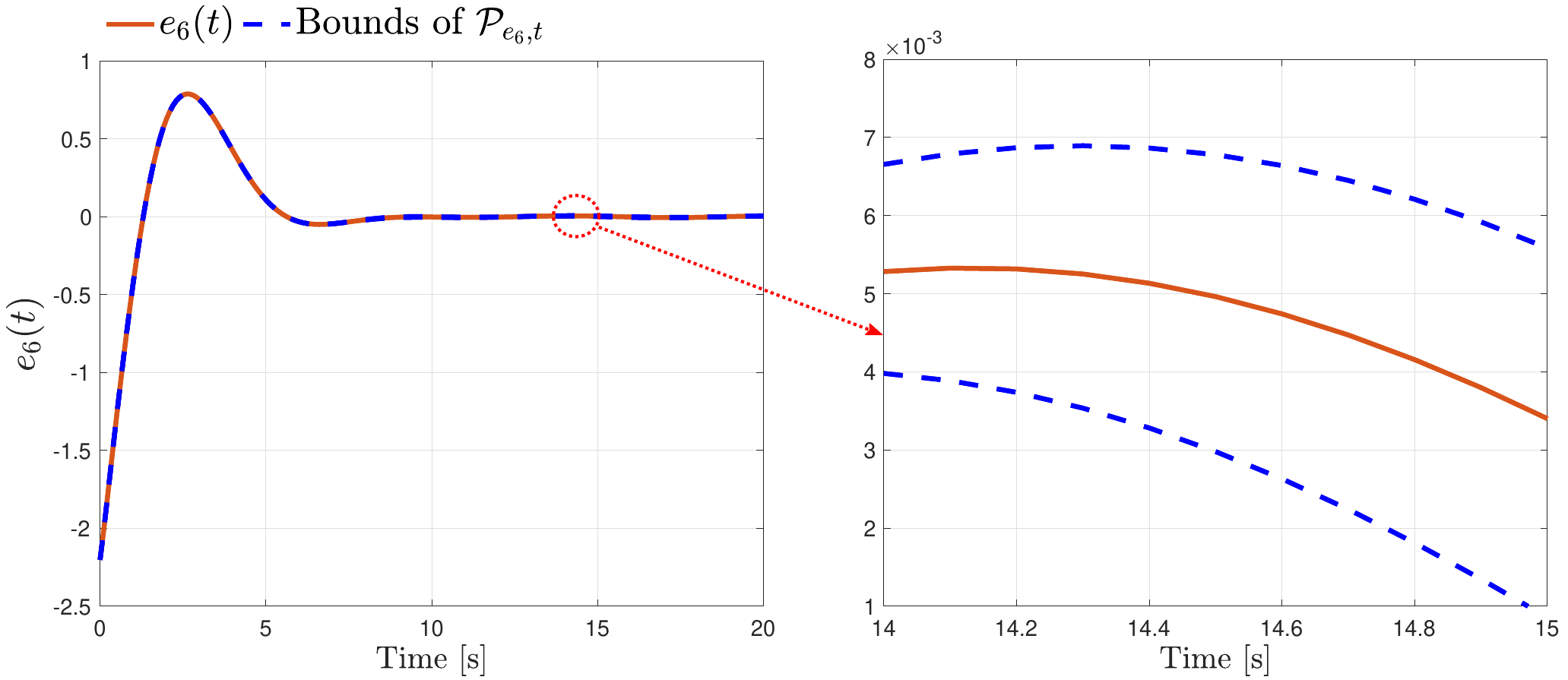}
		\label{fig4a}}\vspace{-4mm}
	\\
	\subfloat[]{\includegraphics[scale=0.25]{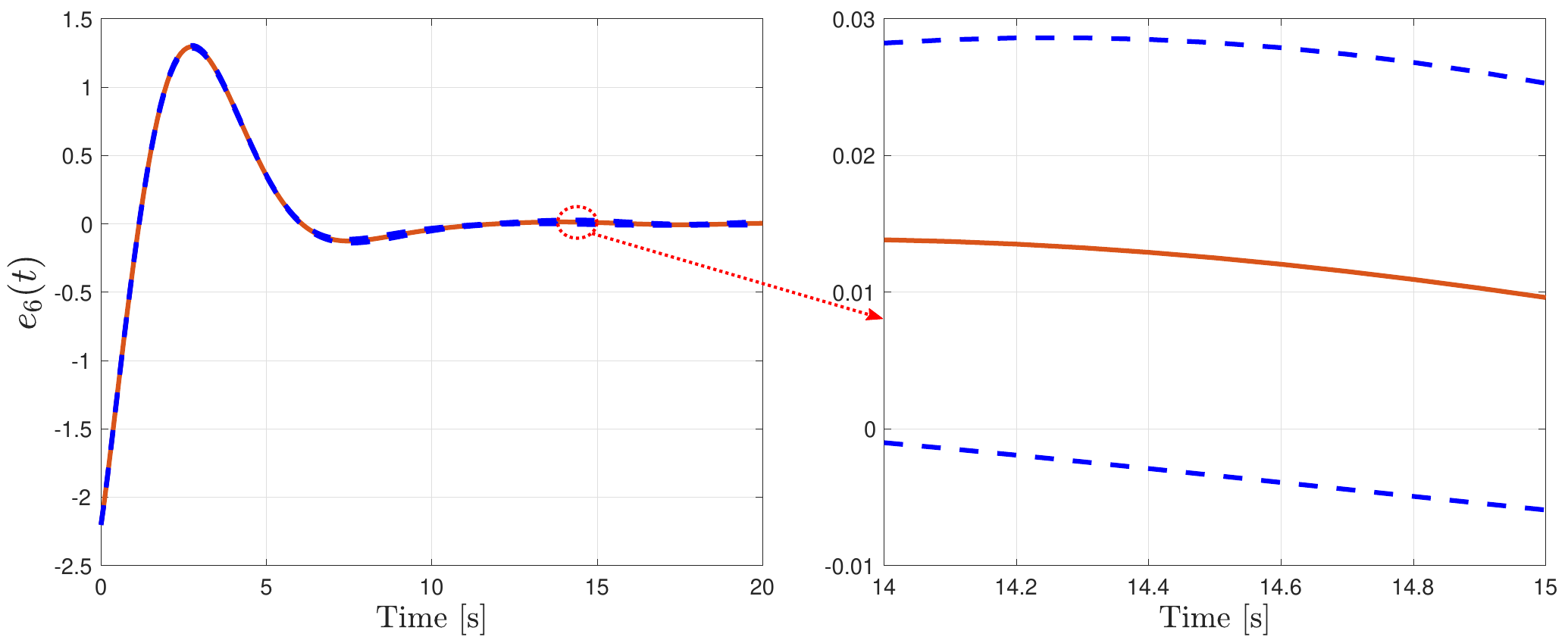}%
		\label{fig4b}}\vspace{-4mm}
	\\
	\subfloat[]{\includegraphics[scale=0.25]{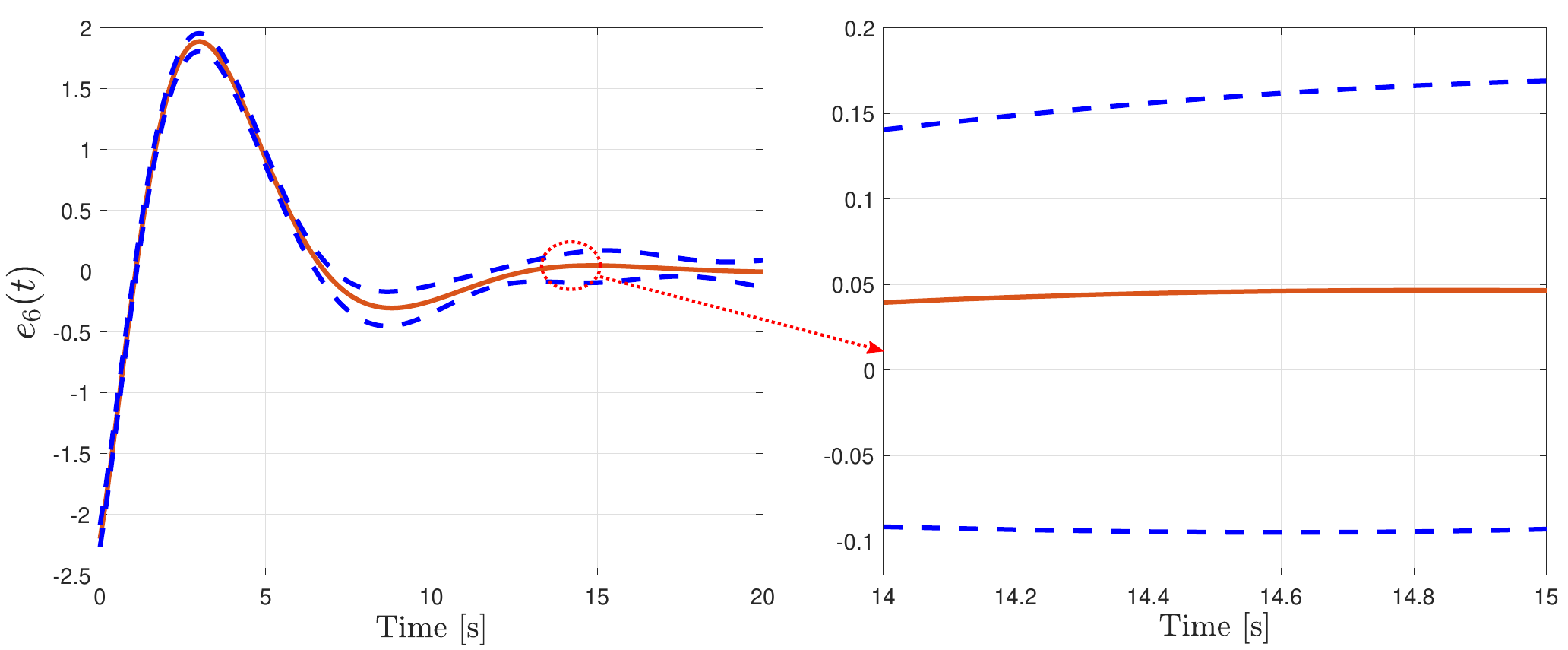}%
		\label{fig4c}}
	\caption{Tracking error $e_6(t)$ and bounds of $\mathcal{P}_{e_6,t}$ under data-driven control with different noise levels: (a) $\bar{w}_i(t)=\bar{v}_i(t)=0.001$; (b) $\bar{w}_i(t)=\bar{v}_i(t)=0.01$; (c) $\bar{w}_i(t)=\bar{v}_i(t)=0.1$.}
	\label{fig:bound}
\end{figure}

\section{Conclusions}
\label{section5}

In conclusion, this paper has presented a data-driven polytopic approach for output synchronization of unknown heterogeneous MASs, utilizing data instead of explicit knowledge of each follower's dynamics model. The proposed method offers a certified data-driven feedback control protocol that can handle perturbed offline data and uncertainties in the system matrices. By means of a unique data-based polytopic representation of the MASs, an approximate solution of the output regulator equations and a stabilizing control gain are obtained. The stability of the tracking error polytope is ensured, and sufficient data-based conditions for near-optimal output synchronization are provided. Future research directions could explore less conservative approximations of the tracking error polytope and investigate output synchronization under event-triggered control. These advancements would further enhance the performance and applicability of the data-driven approach in real-world scenarios.



\bibliographystyle{IEEEtran}
\bibliography{paper}

\end{document}